\documentclass[11pt]{article}
\usepackage{graphicx} 
\usepackage{fullpage}
\usepackage{amsfonts}
\usepackage{amsmath}
\usepackage{xcolor}
\usepackage[colorlinks=true, citecolor=blue]{hyperref}

\usepackage{float}
\usepackage{tcolorbox}
\usepackage{mathtools}
\usepackage{lipsum}
\usepackage{blindtext}
\usepackage{titlesec}
\usepackage{amssymb}
\usepackage{amsthm}
\usepackage{rotating}
\usepackage{subcaption}
\usepackage{caption}
\usepackage{svg}
\usepackage{verbatim}
\usepackage{enumerate}
\usepackage{float}
\usepackage{natbib}

\usepackage{ifthen}
\usepackage{color}
\usepackage{color-edits}
\addauthor{cp}{red}
\addauthor{jz}{blue}

\newcommand{\footremember}[2]{%
    \footnote{#2}
    \newcounter{#1}
    \setcounter{#1}{\value{footnote}}%
}

\title{Incentivizing Desirable Effort Profiles in Strategic Classification: The Role of Causality and Uncertainty}
\author{Valia Efthymiou\footremember{alley}{Massachussetts Institute of Technology. Email: valia554@mit.edu} 
\and Chara Podimata\footremember{somelse}{Massachussetts Institute of Technology. Email: podimata@mit.edu}
\and Diptangshu Sen\footremember{trailer}{Georgia Institute of Technology. Email: dsen30@gatech.edu}
\and Juba Ziani\footremember{somethingelse}{Georgia Institute of Technology. Email: jziani3@gatech.edu}
}

\begin{document}

\maketitle

\newcommand{\cF}{\mathcal{F}}
\newcommand{\cC}{\mathcal{C}}
\newcommand{\cP}{\mathcal{P}}
\newcommand{\cG}{\mathcal{G}}
\newcommand{\cN}{\mathcal{N}}
\newcommand{\cA}{\mathcal{A}}
\newcommand{\cH}{\mathcal{H}}
\newcommand{\cZ}{\mathcal{Z}}

\newcommand{\cI}{\mathcal{I}}
\newcommand{\R}{\mathbb{R}}
\newcommand{\E}{\mathbb{E}}
\newcommand{\bP}{\mathbb{P}} 
\newcommand{\calD}{\mathcal{D}}
\newcommand{\eps}{\varepsilon}
\newcommand{\Dxe}{\Delta x_{e}}
\newcommand{\calY}{\mathcal{Y}}

\newcommand{\caus}{\mathcal{C}}
\newcommand{\prox}{\mathcal{P}}

\newtheorem{obs}{Observation}
\newtheorem*{blob}{Observation}
\newtheorem{lem}{Lemma}
\newtheorem{rem}{Remark}
\newtheorem{clm}{Claim}
\newtheorem{conj}{Conjecture}
\newtheorem{cor}{Corollary}
\newtheorem{defn}{Definition}
\newtheorem{aspt}{Assumption}
\newtheorem{prop}{Proposition}
\newtheorem{thm}{Theorem}
\newtheorem{remark}{Remark}

\definecolor{ao(english)}{rgb}{0.0, 0.5, 0.0}
\newcommand{\myworries}[1]{\textcolor{red}{#1}}
\newcommand{\ds}[1]{\textcolor{blue}{[Sen: #1]}}
\newcommand{\ve}[1]{\textcolor{ao(english)}{[Valia: #1]}}
\newcommand{\jz}[1]{\textcolor{orange}{[Juba: #1]}}

\newcommand{\des}{\mathcal{D}}
\newcommand{\und}{\mathcal{U}}
\renewcommand{\th}{\widetilde{h}}
\newcommand{\eff}{e}
\newcommand{\effopt}{e^{\star}}
\newcommand{\Cost}{\textsf{Cost}}
\newcommand\numberthis{\addtocounter{equation}{1}\tag{\theequation}}
\newcommand{\bbR}{\mathbb{R}}
\newcommand{\C}{\mathbb{C}}
\newcommand{\I}{\mathbb{I}}
\newcommand{\z}{z}

\newcommand{\bh}{\bar{h}}
\newcommand{\bw}{\bar{w}}
\newcommand{\fstar}{f^\star}
\newcommand{\pistar}{\pi^\star}

\newcommand{\xhdr}[1]{\vspace{2mm} \noindent{\bf #1}}

\begin{abstract}
We study strategic classification in binary decision-making settings where agents can modify their features in order to improve their classification outcomes. Importantly, our work considers the causal structure across different features, acknowledging that effort in a given feature may affect other features. The main goal of our work is to understand \emph{when and how much agent effort is invested towards desirable features}, and how this is influenced by the deployed classifier, the causal structure of the agent's features, their ability to modify them,  and the information available to the agent about the classifier and the feature causal graph.

In the complete information case, when agents know the classifier and the causal structure of the problem, we derive conditions ensuring that rational agents focus on features favored by the principal. We show that designing classifiers to induce desirable behavior is generally non-convex, though tractable in special cases. We also extend our analysis to settings where agents have incomplete information about the classifier or the causal graph. While optimal effort selection is again a non-convex problem under general uncertainty, we highlight special cases of partial uncertainty where this selection problem becomes tractable. Our results indicate that uncertainty drives agents to favor features with higher expected importance and lower variance, potentially misaligning with principal preferences. Finally, numerical experiments based on a cardiovascular disease risk study illustrate how to incentivize desirable modifications under uncertainty. 
\end{abstract}

\section{Introduction}\label{sec:intro}

The widespread adoption of automated decision-making systems has brought significant attention to the issue of \emph{strategic classification}---a machine learning setting where individuals modify their features to secure favorable outcomes. This phenomenon is common in many domains: students enroll in preparatory courses to enhance their chances at college admission; job seekers tailor their resumes to align them with AI-based hiring algorithms, and individuals adjust their financial behaviors to improve credit scores. Some of these modifications reflect \emph{genuine efforts} to enhance one's qualifications or financial responsibility (e.g., acquiring new skills or consistently paying off loans), while others \emph{effectively game} the system, e.g., artificially boosting credit scores by opening new credit lines or strategically targeting specific keywords in algorithmic resume screening.

The distinction between \emph{desirable} and \emph{undesirable} modifications is not always clear-cut. While gaming is typically regarded as problematic, even genuine improvements can vary in how desirable they are. For example, in healthcare, encouraging patients to adopt preventive lifestyle changes (such as improved diet and regular exercise) may be preferable to medical interventions like medication or surgery for conditions such as obesity or hyperlipidemia. This highlights the nuanced nature of strategic classification: interventions that lead to real improvements may still not align with preferred or \emph{desirable} forms of improvement, where desirability is decided by the learner.

Further, a key challenge in strategic classification is that features are often interdependent. That is, modifications to one feature can have cascading effects on others. For example, increasing the number of credit cards an individual holds will also lower their credit utilization percentage, indirectly influencing their credit-worthiness. Similarly, reducing alcohol consumption or improving dietary habits can mitigate multiple health risks, such as obesity, hyperlipidemia, and cardiovascular diseases. These dependencies are best captured using a \emph{causal graph}, a framework that has been explored in a limited amount of prior work~\cite{kleinberg2019classifiers,miller2020strategic,shavit2020incentives,bechavod2022gaming} in the specific context of strategic classification. 

Our work builds upon this causal perspective on strategic classification, investigating how agents respond to decision-making systems and, in particular, when their strategic behavior aligns with desirable modifications. We adopt a framework in which a principal (e.g., a decision-maker or machine learning classifier) deploys a model, and agents (or individuals) strategically adjust their features to maximize their probability of receiving a favorable classification outcome. 

\vspace{5pt}

\noindent {{\bf Summary of Contributions.}} Our contributions are as follows: 

\emph{Our Model:}  In Section~\ref{sec:model}, we introduce our model to study incentivizing desirable efforts in the context of causal strategic classification. We build on previous work in two different ways: i) first, we introduce incomplete information to the study of causality in strategic classification, highlighting situations where an agent may not know the classifier, the causal graph, or both; and ii) we introduce a new notion of $\beta$-desirability which quantifies the extent to which agents invest effort in features deemed desirable by the principal.

\emph{Complete Information:}  In Section~\ref{sec:complete}, assuming agents have full knowledge of the classifier and the causal structure, we characterize the optimal effort profiles under various modification cost structures. We establish theoretical conditions guaranteeing investment in desirable effort profiles by rational agents. We also demonstrate that finding classifiers that induce desirable behavior is, in general, a non-convex problem.  
However, we show that when the principal chooses only \emph{one} desirable feature to incentivize, the problem of finding good classifiers becomes convex. We also provide a simple convexification heuristic  for when the number of desirable features is more than one,  ensuring that chosen classifiers do \emph{not} incentivize more than a certain amount of undesirable feature effort.

\emph{Incomplete Information:}  We extend our analysis to the setting where agents lack information about either the classifier or the causal graph (or both) in Section~\ref{sec:incomplete}. There, incomplete information is modeled as agents having Gaussian priors about the classifier and the causal graph. First, we show that in presence of uncertainty over both the classifier and the causal graph, choosing how to invest effort optimally is a non-convex problem for the agent. However, the problem becomes tractable under \emph{partial uncertainty}, and we provide a semi-closed-form characterization of optimal effort profiles for agents in partial uncertainty settings. %

\emph{Case study:} Finally, in Section~\ref{sec:experiments}, we complement our theoretical insights in the incomplete information setting with numerical experiments, basing our experimental setup on a medical study from previous work that predicts \emph{risk of cardiovascular disease} (CVDs) in adults. In the process, we provide insights into how to incentivize changes in desirable features under uncertainty.

\vspace{5pt}
\noindent {\bf Related Work:} 
Strategic classification, a machine learning setting where agents can manipulate or modify their own features to improve their outcomes, has been widely studied in recent years under a range of assumptions. This belongs to a broad class of problems in economics called principal-agent problems where agents act strategically in their self-interests which are often misaligned with the principal's interests~\cite{grossman1992analysis,ross1973economic,laffont2009theory,sappington1991incentives}. Early works in strategic classification focused on scenarios where agents manipulate their observable features solely to ``game'' a published classifier, thereby increasing their chances of a favorable label without genuinely improving underlying attributes (e.g., \cite{hardt2016strategic, braverman_randomness, dong2017strategicclassificationrevealedpreferences,zhang2022fairness,lechner2023strategic,chen2020learning,ahmadi2021strategic,sundaram2023pac}). Over time, the literature expanded to consider settings where agents make substantive changes to their features (i.e., effectively investing in real improvements) rather than relying on superficial modifications (e.g., \cite{kleinberg2019classifiers,bechavod2022information,bechavod2022gaming,shavit2020incentives,harris2021stateful}). In many cases, actual improvement involves investing effort which is not directly observable by the principal --- this again has similarities to the notion of moral hazard in insurance markets~\cite{pauly1968economics,arrow1968economics} and other general settings~\cite{arrow1978uncertainty}. There has also been interest in fairness in strategic classification (e.g., \cite{milli2019social,hu2019disparate,estornell2023group}) but this line of work is more distantly related to ours.

A useful tool to model manipulations as opposed to effective improvement is \emph{causality}. Causal modeling has been extensively studied in decision-making and machine learning, starting with~\cite{pearl2000causality}; see~\cite{kaddour2022causalmachinelearningsurvey} for a recent survey of causal machine learning. 
In the context of strategic classification, a few recent studies have incorporated causal modeling to account for interdependencies among features~\cite{kleinberg2019classifiers,shavit2020incentives,bechavod2022gaming, blum_game_improve,miller2020strategic,horowitz2023causalstrategicclassificationtale}.~\cite{miller2020strategic} highlights that in strategic classification, learning a classifier that incentivizes gaming as opposed to improvement is as hard as learning the underlying causal graph between features.
~\cite{shavit2020incentives} and~\cite{bechavod2022gaming} both focus on special cases of linear causal graphs, unlike our work that considers general cases of linear graphs.~\cite{blum_game_improve} explore strategic classification using a different structural framework known as ``manipulation graphs,'' where each agent has a fixed set of costly effort profiles that they may choose from, defining a bipartite graph between initial agent features and induced features after exerting effort---this is, again, a special case of causal graphs.~\cite{horowitz2023causalstrategicclassificationtale}, similarly to us, distinguish among causal, non-causal (or ``correlative''), and unobserved features. However, they focus on a different objective of maximizing predictive accuracy, we are interested in incentivizing ``desirable'' modifications.

Perhaps closest to our work is the work of~\cite{kleinberg2019classifiers}. Like us, they focus on general causal graphs; however, we highlight several major differences. First, we focus on \emph{classification} settings, while~\cite{kleinberg2019classifiers} focus on regression and scoring settings. Second, we highlight differences in our agent model, where our agents invest effort to pass the classifier with reasonably high probability, while agents in~\cite{kleinberg2019classifiers} \emph{always} exert effort to improve their score. Third, we note that our cost model is strictly more general: where~\cite{kleinberg2019classifiers} focuses on linear costs, our work considers general $\ell_p$-costs. Our results show that this choice of cost is important, noting a sharp distinction in agent behavior between the cases of $\ell_1$-cost and $\ell_p$ costs for $p > 1$. Finally, unlike \cite{kleinberg2019classifiers}, our study incorporates \emph{incomplete information}, where agents may not fully understand either the causal graph or the deployed classifier.

A closely related line of work investigates strategic classification under varying models of \emph{information} available to agents. In many real-world settings, agents may have \emph{incomplete information} about the classifier---either because it is too complex, or because the learner's model is proprietary, or the causal relationships governing feature interactions~\cite{bechavod2022information,cohen2024bayesianstrategicclassification,ghalme2021strategic}. Or, agents might misperceive the classifier due to behavioral biases~\cite{ebrahimi2024double}.
However, we are not aware of any work studying uncertainty on causal graphs, and to the best of our knowledge, we are the first work in the space of strategic classification to incorporate \emph{both} causal modeling and incomplete information in strategic classification.

\section{Model}\label{sec:model}
We consider a \emph{binary classification} problem, where there is an interaction between a principal and an agent. The principal, also known as the \emph{learner}, deploys a machine learning model or classifier. Then, the model assigns a (binary) score or a decision to the agents, based on their \emph{features}. Finally, agents \emph{respond} to the deployed classifier, potentially changing their features to obtain better outcomes, at a cost. We provide a detailed description of the model below.

\subsection{Agent Model: Features and Utilities}

Formally, let $\cF \in \mathbb{R}^d$ be the set of all features and $\calY = \{-1, +1\}$ be the set of labels. Each agent $k$ is defined by a pair $(x_k,y_k)$, where $x_k \in \cF$ is the agent's feature vector and $y_k$ is the agent's true label (i.e., their actual fit or qualification for the classification task at hand).

\xhdr{Causal modeling of feature interactions.} We adopt a \emph{causal} perspective on strategic classification, where different features can impact each other---i.e., a change in a feature $i$ (e.g., alcohol consumption or diet) that has a causal relationship with feature $j$ (e.g., cholesterol) will also induce a change in feature $j$. 

The chain of causality between the different features can be captured using a weighted directed graph $\cG = (\cF, \cA, w)$. We will henceforth call graph $\cG$ the \emph{causal graph}. We slightly abuse notation here and denote the set of nodes in the graph also by $\cF$ to indicate that each feature in $\cF$ corresponds to a node on $\cG$. $\cA$ represents the set of directed edges on $\cG$, where an edge from features $i$ to $j$ indicates that $i$ is causal for $j$. Finally, $w: \cA \rightarrow \R$ captures the weights of the edges. We make no assumption on the structure of $\cG$, other than the fact that it is a directed \emph{acyclic} graph.
\footnote{This is a standard assumption in the causal strategic classification~\cite{miller2020strategic,kleinberg2019classifiers} }
We represent all necessary information about the graph using an \emph{adjacency matrix} $A \in \R^{d \times d}$:

\begin{defn}[Adjacency Matrix of a Causal Graph]
\begin{align}\label{eq:adj_matrix}
       A_{ij} = \begin{cases}
       0, \quad \quad \quad \, \text{if $a_{ij} \notin \cA$,}\\
       w(a_{ij}), \quad \text{if $a_{ij} \in \cA$.}
       \end{cases}
\end{align}
\end{defn}

Here, if there is an edge $a_{ij} \in \cG$, then feature $i \in \cF$ \emph{causally affects} feature $j \in \cF$; in other words, by changing feature $i$, an agent is making an implicit change on feature $j$. Finally, the weight of edge $a_{ij}$, given by $ w(a_{ij})$ indicates that if the value of feature $i$ improves by a unit amount, then the value of the downstream feature $j$ will improve by $w(a_{ij})$\footnote{Throughout the paper, we assume that causal relationships are linear. This is another common assumption in the literature~\cite{kleinberg2019classifiers,shavit2020incentives}}. Importantly, \textit{edge weights can be negative}---an increase in causal feature $i$ might lead to a decrease in feature $j$. 

\paragraph{Desirable vs Undesirable Features} Importantly, we assume that the set of features $\cF$ is divided by the principal into two kinds of features: \textbf{desirable} features and \textbf{undesirable} features denoted by sets $\des$ 
and $\und$ 
respectively. Roughly speaking, \emph{desirable} features are those that the principal wants to incentivize the agent to change; e.g., in the health application of the Introduction, ``alcohol consumption'' would be a \emph{desirable} feature that the principal (for example, the agent's primary care physician) would like to see changed\footnote{Here, we would like to see the value of the feature lowered. Note that whether we want to increase or decrease the value of a feature has no bearing on whether it is desirable; only the fact that this is one of the target features we aim to change is.}. Undesirable features, on the other hand, can be considered as features that 
we would like to disincentivize agents from modifying directly: e.g., directly intervening to lower an agent's cholesterol level via medication such as statins may be less desirable than promoting lifestyle changes (lower alcohol consumption, improved diet, etc.) that will also lower their cholesterol.

\begin{remark}[Relationship between causality and desirability ]
In the traditional causality literature on strategic classification, features on a causal graph are usually classified into \textit{causal} and \textit{non-causal} or \emph{proxy} features. Causal features are those which affect some specified output variable of interest while proxy features are those which do not. A natural question is whether and how desirability of a feature relates with its causality. 

Proxy features are generally seen as undesirable to modify: this is because they do not change the root cause behind an agent's label, hence do not lead to true improvements in said label---this is often referred to as ``gaming'' the classifier. However, causal features may still be undesirable, even if they do not lead to gaming. For example, in our healthcare example, both diet and cholesterol levels are causal for predicting the risk of cardio-vascular disease---in particular, diet is directly causal for cholesterol levels, and cholesterol levels are directly causal for cardio-vascular disease. Yet, it may be preferable to incentivize sustainable interventions such as a better diet (prevention), rather than resorting to short-term fixes like cholesterol-lowering medications that may have significant side effects (treatment). 
\end{remark}

\subsection{Principal - Agent Interaction}

The principal deploys a \emph{linear classifier} denoted as $h_0 \in \R^{d}$ from its normal vector. Under this classifier, an agent with feature vector $x \in \R^{d}$ is assigned a score of $s(x) = h_0^\top x$. There is a pre-determined threshold $\tau \in \R$ and the classification decision $y$ for said agent is given by: 
\begin{align*}
    y(x) = \mathbf{1}\left[s(x) \geq \tau \right]. 
\end{align*}

\xhdr{Agent's Best Response.} The agent is assumed to have Gaussian priors $\Pi_h := \mathcal{N}(\mu_h, \Sigma_h)$ (where $\mu_h \in \R^{|\cF|}, \Sigma_h \in \R^{|\cF|\times |\cF|}$) over the principal's deployed classifier $h_0$ and $\Pi_{\C} := \mathcal{N}(\mu_w, \Sigma_w)$ (where $\mu_w \in \R^{|\cA|}, \Sigma_w \in \R^{|\cA| \times |\cA|}$) over the edge weights of the causal graph $\cG$\footnote{Gaussian priors are frequently used to model incomplete information~\cite{kong2020information,elzayn2019price}}. The agent is always assumed to know the topology of the causal graph. We explore two kinds of information structures: 
\begin{enumerate}
    \item The \textit{Complete Information setting} where the agent fully knows the classifier $h_0$, i.e, $\mu_h = h_0$ and $\Sigma_h = \bf{0}$ and the weights of all edges of $\cG$, i.e., $\mu_w = w$ and $\Sigma_w = \bf{0}$. See Section~\ref{sec:complete} for the complete information setting. 
    \item The \textit{Incomplete Information setting} where i) there is uncertainty over the principal's classifier $h_0$, i.e., $\mu_h$ may differ from $h_0$ (bias) and $\Sigma_h \neq \bf{0}$ (variance), and/or ii) there is uncertainty over the edge weights of the causal graph $\cG$, i.e., $\mu_w$ may differ from $w$ (bias) and $\Sigma_w \neq \bf{0}$ (variance). See Section \ref{sec:incomplete} for the incomplete information setting. 
\end{enumerate}
If $y(x) = 0$, the agent also knows the amount $\alpha > 0$ by which she fell short of passing the classifier; e.g., in a loan approval setting, an agent may know their current credit score and the threshold credit score that the bank uses to decide who gets approved for a loan. They might, however, not fully understand how said credit score is calculated in the first place.

The agent's goal is to obtain a positive classification outcome, i.e. $y(x) = 1$ (``to pass the classifier''). Therefore, she attempts to \emph{change} her feature vector $x$ by investing some \emph{exogenous effort} $\eff \in \R^{|\cF|}$, which we call the agent's \emph{exogenous effort profile}. Importantly, this exogenous effort profile, in our model, is exerted directly on features\footnote{This is without loss of generality, by a simple revelation principle type of argument}: i.e., $e$ is a vector quantifying how much an agent changes their features \emph{directly}. However, remember that exerting exogeneous effort on a subset of the features $\cF$ can also lead to features, particularly those that no effort was exerted on, to change \emph{endogenously}, due to causality. We call this the \emph{induced} or \emph{endogenous feature change}.

Now, suppose the agent's modified feature vector after investing effort is given by $x'(\eff)$. We define the net change in features due to effort $\eff$ as:
\[
      \Delta x(\eff) = x'(\eff) - x,
\]
where $\Delta x(\eff)$ can be computed using the structure of the causal graph $\cG$. This effort comes at a cost, modeled through a \emph{cost function} 
\[
\Cost: \R^{d} \to \R_{\geq 0},
\]
where $\Cost(e)$ is the cost incurred to perform effort $\eff$. We mainly focus on (weighted) $\ell_p$-norm cost functions for all $p \geq 1$. Formally, 
\begin{align}\label{eq:cost}
             &\textsf{Cost}(\eff) = \left(\sum_{f \in \cF} c_f\left \vert \eff_f \right \vert^p \right)^{1/p},~\text{where}~c_f > 0~~\forall f \in \cF,
\end{align}
where $c_f$ represents a cost multiplier associated with investing unit exogenous effort into feature $f$.

To express $\Delta x(\eff)$ in terms of the causal graph $\cG$, we define the \emph{contribution matrix} of $\cG$, which allows us to quantify how much effort profile $\eff$ maps to a total change in features, including both exogenous effort and induced feature changes:

\begin{defn}\label{def:cont_matrix}
The contribution matrix $\C \in \R^{d \times d}$ associated with causal graph $\cG$ is: 
\begin{align*}
      &\C_{ii} = 1 \quad \forall~i \in [d], \quad \text{and}\\
      &\C_{ij} = \sum_{p \in \mathcal{P}_{ij}} \omega(p) \quad \forall~ i,j \in [d],~i \neq j,
\end{align*}
where $\mathcal{P}_{ij}$ is the set of all directed paths from node $i$ to node $j$ on $\cG$ and $\omega(p)$ is the weight of path $p \in \mathcal{P}_{ij}$ with $\omega(p) = \prod_{a \in \cA, a \subset p}w(a)$. 
\end{defn}

Note that in causal graphs, a given feature $i$ may affect another feature $j$ directly---in which case there is an edge of non-zero weight from $i$ towards $j$, but also indirectly, through other features. For example, there may be a path from feature $i$ to feature $j$ through intermediary features $i_1,\ldots, i_k$, where $i \to i_1 \to i_2 \to \ldots \to i_k \to j$\footnote{$x \to y $ indicates that $x$ is directly causal for $y$.}. Therefore, investing effort in feature $i$ will lead to modifications of not only feature $i$, but also all features $i_1, \ldots, i_k,~j$ that it is directly or indirectly causal for. The contribution matrix $\C$ quantifies the impact of any given feature $i$ on any other feature $j$, \emph{even when the features have an indirect causal relationship}.

\begin{obs}\label{obs:comoute_contri}
Given matrix $A$ as defined in Equation~\eqref{eq:adj_matrix}, the contribution matrix is given by 
\[
\C = \sum_{k=0}^{|\cF|} A^k,
\]
and therefore can be computed in polynomial time in $|\cF|$. Intuitively, $A^k$ represents the contribution of all paths of size exactly $k$; because our graph is acyclic, the longest path must have length at most $|\cF|$. 
\end{obs}
\begin{proof}
This is a well-known result, but we provide a proof in Appendix~\ref{sec:app_sup} for completeness.
\end{proof}
\noindent 
Given the contribution matrix $\C$ and exogenous effort profile $\eff$, the net change in features $\Delta x(\eff)$ will be given by:
\begin{align}\label{eq:delta_x}
       \Delta x(\eff) = \C^\top \eff.
\end{align}
To simplify notation, we drop the dependence on $e$ and write $\Delta x$ when clear from context. We demonstrate how to construct $\C$ and how to compute $\Delta x$ given an effort profile $\eff$ through a toy example in Figure~\ref{fig:schematic_C}.
\begin{figure}[hbt!]
    \centering
    \includegraphics[width=.85\textwidth]{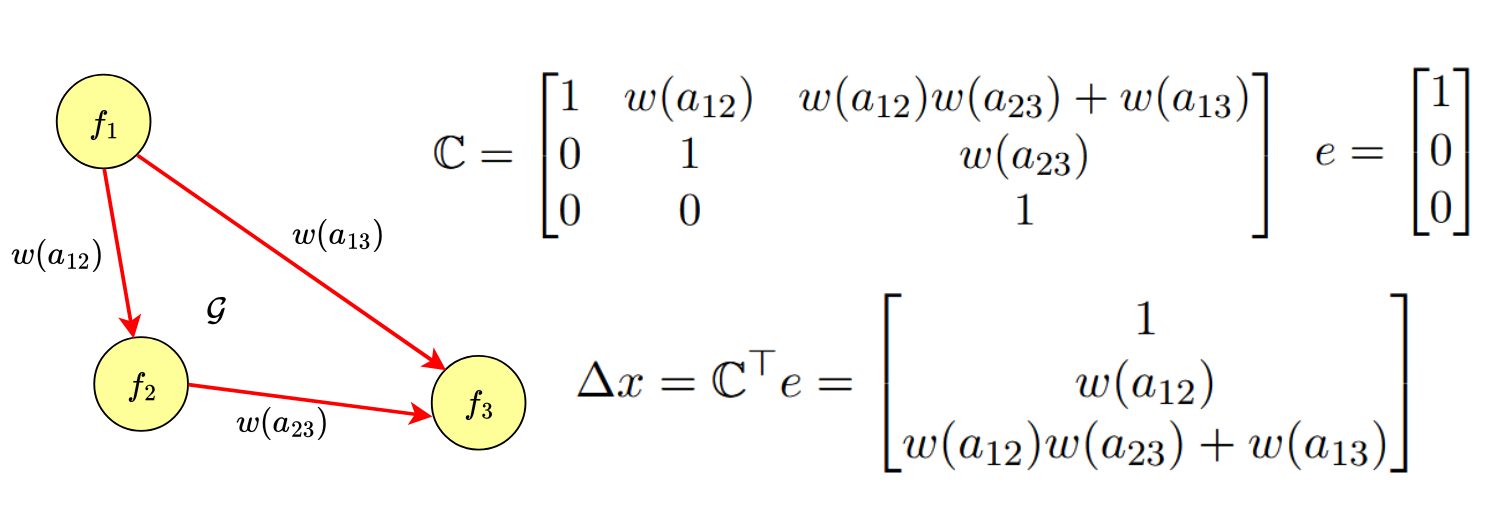}
    \caption{Consider a simple causal graph $\cG$ with $|\cF| = 3$. Feature $1$ directly affects features $2$ and $3$ and feature $2$ directly affects feature $3$. Feature $1$ also indirectly affects feature $3$ through the path $1 \to 2 \to 3$. The contribution matrix $\C$ captures both of these effects.}
    \label{fig:schematic_C}
\end{figure}

The agent chooses their optimal effort profile $\effopt$ that ensures that $y(x'(\effopt)) = +1$ with probability at least $1-\delta$, while incurring the minimum possible cost. We call the effort profile $\effopt(\Pi_h, \Pi_{\C})$ the agent's \emph{best response} to priors ($\Pi_h, \Pi_{\C}$). Formally: 
\begin{align*}
   \effopt (\Pi_h, \Pi_{\C}) = \arg \min_{\eff} \quad &\Cost(\eff) \\
    \text{s.t.} \quad &\bP_{h \sim \Pi_h, \C \sim \Pi_{\C}}\left[h^{\top}\C^{\top}\eff \geq \alpha \right] \geq 1-\delta \numberthis{\label{opt:agent_br}}.
\end{align*}
Note that the constraint ensures that an agent passes the classifier with probability at least $1-\delta$, measured with respect to their prior on the causal graph and on the classifier. In particular, if they exert effort profile $\eff$, their features change by $\C^\top \eff$, so their score changes by $h^{\top}\C^{\top}\eff$, and they have to improve their score by $\alpha$ to make a positive classification outcome. 

\subsection{Incentivizing Effort towards Desirable Features} 

We are interested in the \emph{properties} of the effort profile that the agent exerts as a result of best-responding to the principal's classifier, and in particular understanding the amount of effort they exert towards desirable features in set $\des$ and undesirable features $\und$. The goal is to promote effort towards desirable features and away from undesirable features, i.e. to understand \textit{when is it in the agent's best interest to invest more effort into desirable versus undesirable features?}. 

To do so, we define a notion of $\beta$-desirability, that measures the ratio of investment in features in $\des$ vs $\und$: 

\begin{defn}[$\beta$-desirable effort profiles]\label{defn:good}
Given $0 < \beta \leq 1$, an exogenous effort profile $\eff$ is said to be $\beta$-``desirable" if: 
\[
    \|\eff_\des\|_2 \geq \beta \|\eff\|_2, 
\]
i.e., the magnitude of effort made towards desirable features is at least a fraction $\beta$ of the magnitude of total effort. 
\end{defn}

From the principal's point of view, incentivizing $\beta$-desirable effort profiles is not straightforward since agents are strategic, and may prefer undesirable features if they are low-cost to manipulate. 

\section{The Complete Information Setting} \label{sec:complete}

In the complete information setting, the agent knows precisely the true classifier $h_0$ deployed by the principal---equivalently, her prior $\Pi_h$ satisfies $\bh = h_0$ (the mean belief matches the true classifier) and the covariance is given by $\Sigma_h = \bf{0}$ (there is no uncertainty). She also fully knows the causal graph $\cG$---i.e., $\bw = w$ and $\Sigma_w = \bf{0}$. Therefore, in the complete information case, the deterministic tuple ($h_0, \C$) is enough to characterize agent beliefs. 

When the agent has no uncertainty about either the classifier or the causal graph, the agent's optimization problem ~\eqref{opt:agent_br} can be written as:
\begin{align*}
   \effopt (h_0, \C) = \arg \min_{\eff} &\quad \Cost(\eff) \\
    \text{s.t.}~~&\quad (\C h_0)^{\top}\eff \geq \alpha  \numberthis{\label{opt:agent_fullinfo}}.
\end{align*}
In other words, the agent must find the minimum-cost effort profile that passes the (known) classifier $h_0$. We first prove a technical result that helps further simplify the complete information setting. Roughly speaking, the proposition states that for the complete information case, it suffices to only focus on non-negative efforts for all features.

\begin{prop}\label{prop:y_pos}
For the cost functions defined in \eqref{eq:cost}, we can assume that: 
$(\C h_0)_f \geq 0$ and that $(e)_f \geq 0~~\forall f \in \cF$,  without loss of generality.
\end{prop}

The proof of the proposition can be found in Appendix~\ref{app:y_pos}.

\paragraph{Computation of the Optimal Effort Profile.} Using Proposition \ref{prop:y_pos}, we can rewrite optimization problem~(\ref{opt:agent_fullinfo}) under the complete information setting as follows: 
\begin{align}\label{opt:full_info_final}
     \effopt(h_0, \C) = \arg \min_{\eff \geq 0} \quad \textsf{Cost}(\eff)  \quad s.t. \quad (\C h_0)^{\top}\eff \geq \alpha, 
\end{align}
where $\C h_0 \geq 0$. Program~(\ref{opt:full_info_final}) is a convex optimization problem: the objective is convex for our cost functions, and all constraints are linear. As such, this program can be solved efficiently. 

\subsection{Characterization of Optimal Effort Profiles}

We now investigate the structural properties of the optimal effort profile $\effopt$ for the cost function in Eq.~\eqref{eq:cost}, individually for the cases of i) $p = 1$ and ii) $p > 1$. 
The contributions of this section are two-fold: first, we show that the \emph{structure} of the optimal effort profile largely depends on the cost function that an agent optimizes over. In particular, we prove that in the case of the $\ell_1$ cost function, agents only invest effort into \emph{one} feature when best-responding to the classifier $h_0$ (Lemma~\ref{lem:linear_onefeature}). Instead, for general $\ell_p$ costs with $p > 1$, the agents' effort profile is \emph{significantly} more diversified and covers any non-trivial dimension i.e., one with contribution $(\C h_0)_f > 0$ (Lemma~\ref{lem:l2_effort}). Second, we derive conditions under which effort profiles that put a significant amount of weight on \emph{desirable} features are incentivized, providing insights as to how to set a classifier $h_0$ to incentivize effort exertion on \emph{desirable} features (Theorems~\ref{thm:l1_good} and ~\ref{thm:l2_good}).

\subsubsection{Case 1 ($\ell_1$-norm costs)} 
For $p=1$, we have the following optimization problem for the agent: 
\begin{align}\label{opt:linear}
     \min_{\eff \geq 0} \quad c^{\top}\eff \quad \text{s.t.} \quad 
     (\C h_0)^{\top}\eff \geq \alpha. 
\end{align}

We first show that the case where $\alpha \leq 0$ is insignificant and hence we will only focus on $\alpha > 0$. Indeed, $\alpha \leq 0$ indicates that the agent has already passed the classifier (therefore has a non-positive distance from the decision boundary) and hence, they should not invest any effort into modifying features. Formally, 
\begin{prop}\label{prop:alpha_neg}
When $\alpha \leq 0$, then $\effopt = 0$ which means that the agent does not need to invest any effort to change her features. 
\end{prop}
\begin{proof}
Note that the objective value is greater than or equal to zero since $c_f > 0$ for all $f \in \cF$ and $\eff \geq 0$. But, since $\alpha \leq 0$, $\eff= 0$ is feasible to the above problem and it also achieves an objective value of $0$. Therefore, $\eff = 0$ must be optimal.  
\end{proof}

\noindent 
We next focus entirely on the case where $\alpha > 0$. Our first result is characterizing the structure of the agent's best response. 

\begin{lem}\label{lem:linear_onefeature}
When $\alpha > 0$, there exists an optimal effort profile for the agent in which she needs to modify \textbf{exactly one feature} to pass the classifier $h_0$. The optimal feature to modify $f^*$ is the one which offers the best ratio of contribution to cost, i.e., 
\[
     \fstar \in \arg\max_{f \in \cF}~\frac{(\C h_0)_f}{c_f}, 
\]
and the optimal amount of effort to be invested into that feature is given by: 
\[
     \eff_{\fstar} = \frac{\alpha}{(\C h_0)_{\fstar}}.
\]
\end{lem}

The proof of the lemma can be found in Appendix~\ref{app:linear_onefeature}.

\paragraph{Conditions for $\beta$-desirability.} The previous lemma provides key insights into the optimal effort profile of an agent. It shows that the agent has to invest effort into a single feature which offers her the best ``bang-per-buck". Let $\mathcal{I}^*$ be the set of all such features, i.e.:
\[
    \cI^* = \left\{\fstar:~\fstar \in \arg\max_{f \in \cF}\frac{(\C h_0)_f}{c_f} \right\}. 
\]
Therefore, by Definition~\ref{defn:good}, if $\cI^* \cap \und = \emptyset$, then the agent's best response is guaranteed to be a $\beta$-desirable effort profile. We now formalize this idea and present the main result of this section:
\begin{thm}\label{thm:l1_good}
If there exists a desirable feature $\fstar$ ($\fstar \in \des$) such that: 
\[
      \max_{f \in \mathcal{\und}} \frac{(\C h_0)_f}{c_f} < \frac{(\C h_0)_{\fstar}}{c_{\fstar}},
\]
then the agent's best response is always a $\beta$-desirable effort profile for any $\beta \in (0, 1]$. 
\end{thm}
\begin{proof}
The proof follows directly from Lemma~\ref{lem:linear_onefeature} and Definition~\ref{defn:good}.
\end{proof}

\subsubsection{Case 2 ($\ell_p$-norm costs for $p > 1$)}

For $p > 1$, the agent is solving the following optimization problem when best-responding:  
\begin{align}\label{opt:q_norm}
    \min_{\eff \geq 0} \quad \left(\sum_{f \in \cF}c_f (\eff_f)^p\right)^{1/p} \quad 
    \text{s.t.} \quad (\C h_0)^{\top}\eff \geq \alpha.
\end{align}
\begin{lem}\label{lem:l2_effort}
When the cost function is the weighted $\ell_p$-norm of the effort for $p > 1$, the optimal effort profile for the agent $\effopt$ satisfies: 
\[
        \effopt_f \propto \left( \frac{(\C h_0)_f}{c_f} \right)^{1/(p-1)} \quad \forall~f \in \cF.
\]
\end{lem}

The proof can be found in Appendix~\ref{app:l2_effort}.

\paragraph{Conditions for $\beta$-desirability} Since we know the structure of the agent's optimal effort profile, we can identify conditions under which the best response is $\beta$-desirable. 
\begin{thm}\label{thm:l2_good}
For a $\ell_p$-norm cost function with $p > 1$, the agent's best response is always a $\beta$-desirable effort profile if:  
\[
    \left[ \sum_{f \in \des} \left( \frac{(\C h_0)_f}{c_f} \right)^{2/(p-1)} \right]^{1/2} \geq \frac{\beta}{\sqrt{1-\beta^2}}  \left[ \sum_{f \in \und} \left( \frac{(\C h_0)_f}{c_f} \right)^{2/(p-1)} \right]^{1/2}
\]
\end{thm}
\noindent
When $c = \bf{1}$ and $p = 2$, the condition reduces to: 
\[
       \lVert(\C h_0)_{\des}\rVert_2 \geq \frac{\beta}{\sqrt{1-\beta^2}}\lVert(\C h_0)_{\und}\rVert_2,
\]
In other words, in the complete information setting and when agents have $\ell_2$ cost functions, if the magnitude of the net contribution per unit cost along the desirable features relative to the undesirable features is high enough, then it is always in the agent's best interest to invest more in desirable features. 

\subsection{Desirable Classifiers and Where to Find Them}

So far, we have provided conditions which help the principal answer the following problem: ``does a given classifier $h_0$ induce strategic agents to invest only in $\beta$-desirable effort profiles?'' This means that given a classifier $h_0$, we can check that the classifier incentivizes a good profile. However, this does not answer the question of developing algorithms for \emph{finding} such a profile. We start this section with a negative result: the set of $\beta$-desirable classifiers, is, in general, non-convex (Lemma~\ref{lem:comp_des_nonconvex})---in particular, there are typically no well-known methods for finding a good effort profile in high dimensions. However, we investigate a simple condition under which the problem becomes convex (Proposition~\ref{prop:convex_p13}), and we also provide a heuristic (using convexification) that ensures that not too much effort is spent on undesirable features (Proposition~\ref{prop:convex_relax}). The benefit of having a convex design space of desirable classifiers is that i) we can find a desirable classifier using standard optimization techniques, and ii) it allows the principal to choose a classifier that \emph{simultaneously} minimizes (convex) accuracy losses and induces desirable effort profiles.

\paragraph{The Space of Desirable Classifiers is Non-Convex.}
Our first main result shows that finding desirable classifiers is a non-convex problem, hence effectively computationally \textit{hard} in general. 

\begin{lem}\label{lem:comp_des_nonconvex}
There exists an instance of the problem $(\C,h_0)$ and $\beta > 0$, where the space of $\beta$-desirable classifiers $\cH$ is non-convex.
\end{lem}

The proof of the lemma is provided in Appendix~\ref{app:comp_des_nonconvex}. In fact, whenever there is more than one desirable feature, i.e., $|\des| > 1$, $\cH$ can be shown to be a non-convex set.

Since the space of desirable classifiers is, in the worst case, non-convex, finding the best classifier in this set (that maximizes some classification accuracy metric) is equivalent to solving a non-convex optimization problem to global optimality, which is typically an NP-hard problem. 

\paragraph{Special Case: When the Learner Focuses on Incentivizing a Single Desirable Feature.} We now highlight a special case where the space of $\beta$-desirable classifiers is \textit{convex}, for any $\beta > 0$ for a specific range of $p$ values. Namely, we focus on the special case of the principal having a \emph{single} desirable feature that they wish to target. Note that this assumption is actually aligned with what we expect to see happening in real life: indeed, a principal may define for themselves which feature they want to incentivize, and focus on one feature where they would really like to see improvements, especially if this is a feature that has historically not been properly leveraged. Not only that, but by targeting a single feature, they may lower the agents' cognitive load for best-responding, which is always desirable in practice. 

\begin{prop}\label{prop:convex_p13}
Suppose that there is only a single desirable feature, i.e., that $|\des| = 1$. Then for any $\beta > 0$, the space of $\beta$-desirable classifiers $\cH$ is convex for any $\ell_p$-norm cost function with $p \in [1, 3]$.
\end{prop}

The proof can be found in Appendix~\ref{app:convex_p13}.

\paragraph{Minimizing Undesirable Features.}  When $|\des| > 1$, we know that in general, the set $\cH$ of $\beta$-desirable classifiers is not convex, and difficult to optimize over. However, we propose a convexification heuristic (parameterized by $\gamma$) where the principal just tries to design a classifier such that ``the total contribution of undesirable features is no more than $\gamma$'':

\begin{prop}\label{prop:convex_relax}
Let $\cH_{w(\und) \leq \gamma} = \{h_0:~ \| (\C h_0)_{\und}\|_{2/(p-1)} \leq \gamma\}$. Then for any $\gamma > 0$, $\cH_{w(\und) \leq \gamma}$ is convex for any $\ell_p$-norm cost function with $p \in [1, 3]$.
\end{prop}

The proof is nearly identical to that of Proposition~\ref{prop:convex_p13} and is omitted to avoid repetition. This result helps the principal guarantee that they can bound the effort exerted on undesirable features, even if they are not able to guarantee a certain target level of $\beta$-desirable effort.

\section{Incomplete Information Setting}\label{sec:incomplete}
In this section, we switch our attention to the \emph{incomplete} information setting. Our first set of results provide a characterization of when the optimal effort profile for agents can be efficiently computed. We highlight that \emph{under partial uncertainty} and for general weighted $\ell_p$ costs, the optimization program of the agent is convex and can be solved efficiently (Lemma~\ref{prop:partial_incomp_convex}). We also show a technical result in Lemma~\ref{lem:delta_charac}, capturing when the above convex program is feasible and when it is not. We conclude by providing a negative result in the setting where there is uncertainty over both the classifier and the causal graph (Proposition~\ref{prop:full_uncert_nonconvex})---in this case, we show that the agent's optimization program is non-convex and hence, \emph{hard to solve} in the worst case. 

We then aim to characterize what the optimal effort profiles look like. In the $\ell_1$ cost case, we show a sharp contrast with the complete information case: namely, an agent may be willing to expend effort across several features, as opposed to a single one in the complete information case (Lemma~\ref{lem:l1_incomp}). In the $\ell_2$ case, we provide a semi-closed form characterization of the agent's optimal effort profile (Theorem~\ref{thm:incomp_l2}). We also highlight how, under some partial information models, it is possible to provide a more interpretable characterization of how the effort per feature depends on $\E[\C h]$ and $\text{Var}(\C h)$ (Corollary~\ref{cor:prop_effort}): the higher the expected importance $(\C h)_i$ of a feature $i$, the more effort is spent towards it, and the higher the variance of $\C h$ towards feature $i$, the less effort is spent towards it. 

Recall that for the incomplete information setting, the agent's optimization problem is: 
\begin{align}\label{opt:agent_br_incomp}
    \effopt(\Pi_h, \Pi_{\C}) &= ~arg\min_{\eff} \quad \textsf{Cost}(\eff) \quad \text{s.t.}\\
    &\bP_{h \sim \Pi_h, \C \sim \Pi_{\C}}\left[ (\C h)^{\top}\eff \geq \alpha \right] \geq 1-\delta, \quad \delta \in (0,1). \nonumber
\end{align}

Below, we briefly remind the reader of the models of information used in this paper.

\paragraph{Models of Information.} We remind the reader here that there can be two different sources of uncertainty: i) the principal's classifier; and, ii)  the edge weights of the causal graph $\cG$ (the graph topology is assumed to be common knowledge). In particular, we have three following combinations of where uncertainty lies in our problem:
\begin{enumerate}
    \item Uncertainty only exists in the principal's classifier, the causal graph is fully known;
    \item Uncertainty only exists in the edge weights of the causal graph, the classifier is fully known; 
    \item Uncertainty exists over both the classifier and the causal graph.
\end{enumerate}
We will henceforth refer to models $1$ and $2$ as models of \textit{partially incomplete information}. 

\subsection{Optimal Effort Computation}

We start by analyzing the computation of an agent's optimal effort in Models $1$ and $2$:

\paragraph{Models $1$ and $2$:} Since uncertainty manifests in the form of Gaussian priors for the agent (as per assumption), it is easy to see that for information models $1$ and $2$ above, \textit{the overall uncertainty is also Gaussian}: i.e., $\C h$ is a Gaussian random variable. In that case, we can rewrite the agent's optimization problem as follows: 
\begin{align}\label{opt:incomp_gaussian}
    \effopt(\Pi_h, \Pi_{\C}) = ~arg\min \quad &\textsf{Cost}(\eff) \\
    \text{s.t.}~~& \bP_{(\C h) \sim \mathcal{N}(\mu_{\C h}, \Sigma_{\C h})} \left[ (\C h)^{\top}\eff \geq \alpha \right] \geq 1-\delta. \nonumber
\end{align}
Our first result is that above problem is a convex optimization problem under Models 1 and 2. 
\begin{lem}\label{lem:partial_incomp_convex}
Under partially incomplete information (models $(1)$ and $(2)$) and cost functions given by Eq.~\eqref{eq:cost}, the agent's optimization problem to find the optimal effort profile $\effopt$, given by \eqref{opt:incomp_gaussian}, is a convex program for any $\delta \leq \frac{1}{2}$. 
\end{lem}

The proof can be found in Appendix~\ref{app:partial_incomp_convex}.

The last result shows that under limited uncertainty, the agent can still efficiently solve for an effort profile that helps her to pass the classifier with high probability. 

\begin{remark}
We note however that this optimization problem \eqref{opt:incomp_gaussian} is not always feasible:  
To see this, take a look at the worst case when an agent has no information about the problem: i.e., $\Sigma_{\C h} \to +\infty$. Then an agent who is acting blind manipulates in a direction that lowers their true score $h_0^\top x$ with probability exactly $1/2$. I.e., with probability at least $1/2$, they never pass the classifier, and a probability of $1-\delta$ with $\delta$ small cannot be guaranteed. As $\delta \to 0$, the uncertainty intuitively makes it impossible for the agent to guarantee that they will pass the classifier, making the problem also infeasible. 
\end{remark}

To further investigate this, we provide a complete characterization of when optimization problem \eqref{opt:incomp_gaussian} is feasible as a function of $\delta$. The following result highlights the trade-off between the degree of uncertainty in the model and the highest coverage probability ($1-\delta$) that can be achieved.
\begin{lem}\label{lem:delta_charac}
Suppose that $\Sigma_{\C h}$ is positive definite. Then the optimization problem \eqref{opt:incomp_gaussian} is:
\begin{enumerate}
    \item feasible when $\delta > \Phi^{-1}\left( - \| \Sigma_{\C h}^{-1/2} \mu_{\C h} \|_2 \right)$, and
    \item infeasible otherwise, 
\end{enumerate}
where $\Phi^{-1}(\cdot)$ indicates the inverse of the standard normal CDF. 
\end{lem}

The proof is given in Appendix~\ref{app:delta_charac}. Here it is also worth pointing out the main difference with the complete information setting --- \emph{with complete information, an agent can always pass the classifier by choosing effort correctly, unlike the incomplete information setting where a positive outcome is not guaranteed.}

\paragraph{Model $3$: }We now show that under model $(3)$ when there is uncertainty over both the classifier and the causal graph, solving for the optimal effort profile is a non-convex problem in general, and classical optimization algorithms cannot be directly used here.

\begin{prop}\label{prop:full_uncert_nonconvex}
Under incomplete information model $(3)$ and cost functions given by Eq.~\eqref{eq:cost}, the agent's optimization problem, given by \eqref{opt:agent_br_incomp}, is a non-convex program.
\end{prop}

The proof can be found in Appendix~\ref{app:full_uncert_nonconvex}. 
Since solving non-convex optimization problems to global optimality can be NP-hard in the worst case, the above result shows that Model $3$ is likely not computationally tractable. 

\subsection{Characterization of Optimal Effort Profiles}\label{subsec:incomp_effort}

Under partially incomplete information (uncertainty over either the principal's classifier or the edge weights of the causal graph, but not both) with Gaussian priors, we have shown that the agent's optimization problem reduces to the following convex program: 
\begin{align}\label{opt:agent_br_convex}
     \effopt(\Pi_h, \Pi_{\C}) = arg\min_{\eff} \quad &\textsf{Cost}(\eff) \\
     \text{s.t.} \quad &\alpha - \mu_{\C h}^{\top}\eff - p_{\delta}\cdot ||\Sigma_{\C h}^{1/2}\eff||_2 \leq 0, \nonumber
\end{align}
where $(\C h) \sim \mathcal{N}(\mu_{\C h}, \Sigma_{\C h})$, $\delta \leq \frac{1}{2}$ and $p_{\delta} = \Phi^{-1}(\delta)$. Our goal is to gain insights into the agent's optimal effort profile for the cost function class outlined in Eq.~\eqref{eq:cost}. While Program~\eqref{opt:agent_br_convex} can be complex and highly non-convex in the general case, we highlight that we can still obtain structural results for $\ell_1$-costs and a semi-closed form characterization for $\ell_2$-costs.

\subsubsection{Case $1$ (Weighted $\ell_1$-norm costs)}
We have the following optimization problem for the agent: 
\begin{align*}
    \min_{\eff} \quad &\sum_{f \in \cF} c_f |\eff_f| \quad \text{s.t.}\\
    & \alpha - \mu_{\C h}^{\top}\eff - p_{\delta}\cdot ||\Sigma_{\C h}^{1/2}\eff||_2 \leq 0.
\end{align*}

Our first result shows that there is a sharp contrast in the structure of the optimal effort profile for $\ell_1$ costs between the complete information setting and partially incomplete information setting.  
\begin{lem}\label{lem:l1_incomp}
Under partially incomplete information, the optimal effort profile $\effopt$ for an agent with weighted $\ell_1$-norm costs requires investment of effort into more than one feature in the worst case.  
\end{lem}

The proof can be found in Appendix~\ref{app:l1_incomp}.

\subsubsection{Case $2$ ($\ell_2$-norm costs)}
We have the following optimization problem for the agent: 
\begin{align*}
    \min_{\eff} \quad &||\eff||_2 \quad \text{s.t.}\\
    & \alpha - \mu_{\C h}^{\top}\eff - p_{\delta}\cdot ||\Sigma_{\C h}^{1/2}\eff||_2 \leq 0.
\end{align*}
\begin{thm}\label{thm:incomp_l2}
The optimal effort profile $\effopt$ for an agent with $\ell_2$-norm cost function under partially incomplete information, is of the following form: 
\[
       \effopt = \lambda^* \left(k_1 I + k_2 \Sigma_{\C h} \right)^{-1}\mu_{\C h},
\]
where $k_1, k_2, \lambda^* > 0$. 
\end{thm}

The full proof can be found in Appendix~\ref{app:incomp_l2}.

\xhdr{An interpretable special case.}
We study a special case of our problem where we can provide an intuitive explanation of how agents decide to exert effort. In particular, we assume the following. 
\begin{aspt}
$\Sigma_{\C h}$ is a diagonal matrix. We denote $(\Sigma_{\C h})_f$ the $f$-th diagonal entry, which corresponds to the uncertainty with respect to the total contribution by feature $f$.
\end{aspt}

We first note that $\Sigma_{\C h}$ being diagonal arises in very natural settings. One such setting is when i) the uncertainty is on the causal graph $\cG$ and ii) the causal graph is bipartite: i.e., features are either \emph{causal} (they affect other features, but cannot be affected themselves) or \emph{proxy} (they are affected by causal features, but cannot affect any other feature). In this case, causal features only have outgoing edges, while proxy features only have incoming edges. Formally: 

\begin{prop}\label{prop:bipartite}
Suppose $\cG$ is a bipartite graph; further, suppose the agent only has uncertainty over the weights of the graph (model $2$), i.e. $\Sigma_h = \bf{0}$. Then, $\Sigma_{\C h}$ is a diagonal matrix.
\end{prop}

The proof can be found in Appendix~\ref{app:bipartite}.

We now highlight our result relating the effort spent on feature $f$ to the total expected contribution of that feature, $(\mu_{\C h})_f$, and the variance of the contribution of said feature, $(\Sigma_{\C h})_f$:

\begin{cor}\label{cor:prop_effort}
If $\Sigma_{\C h}$ is a diagonal matrix with entries $(\Sigma_{\C h})_f$ corresponding to feature $f$, then the optimal effort profile $\effopt$ has the following form: 
\[
          \effopt_f =  \frac{\lambda^*(\mu_{\C h})_f}{k_1 + k_2 \cdot (\Sigma_{\C h})_f} \quad \forall~f \in \cF.
\]
\end{cor}

\noindent 
This result shows that in the optimal effort profile, the agent invests more effort into features that have a higher expected contribution $(\mu_{Ch})_f$. Further, the denominator highlights that agent may shy away from features they have a lot of uncertainty about: a high value of $(\Sigma_{\C h})_f$ leads to a lower effort invested in that feature. 

\subsection{$\beta$-desirability under Incomplete Information} 
We conclude this section with a discussion on how to induce $\beta$-desirable effort profiles under incomplete information. As we see throughout Section~\ref{subsec:incomp_effort}, it may be difficult to characterize the agent's optimal effort in closed form under incomplete information, except for some limited cases. Here, we focus on providing broad insights here and build on this discussion through numerical experiments in Section~\ref{sec:experiments}.  

\paragraph{The Interpretable Special Case: Diagonal Covariance $\Sigma_{\C h}$:}
In this special setting where we have an interpretable form of the agent's optimal effort profile, we can identify conditions that guarantee investment in $\beta$-desirable effort profiles by rational agents. We present the following result:
\begin{cor}\label{corr:beta_model2}
Suppose that the covariance matrix of feature importance, given by $\Sigma_{\C h}$, is a diagonal matrix. In that setting, if all features have the same overall level of uncertainty and the mean feature importance vector $\mu_{\C h}$ satisfies (which follows from Corollary~\ref{cor:prop_effort}): 
\[
      \| \left(\mu_{\C h}\right)_{\des} \|_2 \geq  \frac{\beta}{\sqrt{1-\beta^2}} \| \left(\mu_{\C h}\right)_{\und} \|_2,
\]
then the best response of a rational agent with $\ell_2$-norm cost is to invest in a $\beta$-desirable effort profile. 
\end{cor}

The above result should be intuitive---when agents face the same degree of uncertainty about the importance of all features, they choose which features to invest effort in based on the mean importance of the features. Therefore, it makes sense that the higher the total net importance (measured by the $\ell_2$-norm) of the set of desirable features, higher the incentive for agents to invest in desirable effort profiles. Finally, we relate $\beta$-desirability to the uncertainty on given features: 

\begin{cor}\label{cor:prop_effort_variance}
$\effopt_f =  \frac{\lambda^*(\mu_{\C h})_f}{k_1 + k_2 \cdot (\Sigma_{\C h})_f}$
is decreasing in $(\Sigma_{\C h})_f$.
\end{cor}

In the general case where different features have different levels of uncertainty, if desirable features have a high degree of uncertainty, this pushes agents away from $\beta$-desirable effort profiles. On the other hand, more uncertainty on undesirable features is good for $\beta$-desirability. This is because having a higher degree of uncertainty (higher variance $(\Sigma_{\C h})_f$) about the importance of a feature actively discourages agents from investing effort into said feature.

\paragraph{What does it mean for $\Sigma_{\C h}$ to not be diagonal?} We provide a short discussion of when non-diagonal covariance matrix arise. In particular, non-diagonal covariance matrices $\Sigma_{\C h}$ arise:
\begin{enumerate}
    \item always under Model $1$ (i.e. where the causal graph is fully known, but there is uncertainty over the classifier), and 
    \item under Model $2$ when the causal graph $\cG$ is not bipartite.
\end{enumerate}
We explore the non-diagonal $\Sigma_{\C h}$ case in greater detail in Section~\ref{sec:experiments}, with experiments on real data that consider more general cases where $\Sigma_{\C h}$ may not be diagonal, and in particular covering Model $1$, when the classifier is not fully known to an agent. Our experiments suggest that many of the same insights about $\beta$-desirability hold (even \emph{without} the assumption that $\Sigma_{\C h}$ is diagonal).

\section{Numerical Experiments}\label{sec:experiments}

Our experimental study focuses on a setting where the learner is trying to reduce a population's risk of cardiovascular disease. To do so, we identify relevant features and build a causal graph based on the recent medical study of~\cite{causal_cvd}. Their study aims to identify the causal links between features such as smoking, diet, or obesity, and whether a patient may develop a cardiovascular disease (CVD). The study is based on an expert survey where several experts were asked to quantify whether specific links between two features as well as links between features and the outcome variable (CVD) are \emph{causal} or \emph{merely correlated}. 

\subsection{Experimental Setup}
\paragraph{Identifying Relevant Features} We identify a subset of 8 features used in~\cite{causal_cvd} that we focus on in our experimental study. In particular, we did not include features that cannot be changed such as age or ethnicity, and only include the features that can be modified by an individual. The 8 features we identified are: alcohol consumption, diet, physical activity, smoking, diabetes mellitus (DM), hyperlypidia (HPL), hypertension (HPT), and obesity. We normalized features to be between $[0,1]$\footnote{For simplicity and wlog, we assume that $0$ is the least ``healthy'' value of the feature, and $1$ is the ``healthiest'' value of the feature. For example, for smoking, $1$ maps to not smoking; for activity, $1$ maps to high amount of weekly physical activity.}.

Among these features, and as noted in our introduction, a principal (i.e., a doctor) would like to incentivize people to focus on preventative, lifestyle interventions over medical treatment interventions. Hence, we separate them to desirable and undesirable to modify as follows: 
\begin{itemize}
\item \emph{Desirable:} Alcohol, Diet, Activity, Smoking. Note that desirable means here that these features are desirable to \emph{modify}, not that, for example, alcohol consumption is desirable. 
\item \emph{Undesirable:} DM, HPL, HPT, Obesity. Note that these features are not ``undesirable'' per se, but rather less desirable than lifestyle interventions.
\end{itemize}

\paragraph{Building the Causal Graph:} The study of~\cite{causal_cvd} asked the experts to report the likelihood of causation on a Likert scale from $1$ to $7$, which is then transformed into  ``fuzzy score'' via the Fuzzy Delphi Method~\cite{linstone1975delphi}. We denote this score $s$. A fuzzy score of $0.5$ and above indicates that experts at least moderately agree with a relationship being causal, with an increasing score $s$ indicating stronger agreement. A fuzzy score of $0.5$ or below indicates that the experts at best disagree with the feature being causal, with the strength of the disagreement increasing as the score goes down. 

We follow the expert agreement of~\cite{causal_cvd} to build our causal links. Specifically, we identify a link as causal if and only if $s > 0.5$. Further, since a score of $0.5$ denotes that experts are at the boundary of agreeing vs disagreeing on causality, we renormalize our scores to be between $0$ and $1$: to do so, we apply a linear transformation that maps $s = 0.5$ to a causal weight of $0$, and $s = 1$ to a causal weight of $1$. We obtain the following graph (Figure~\ref{fig:causal_graph}): 
\begin{figure}[hbt!]
    \centering
    \includegraphics[width=.5\textwidth]{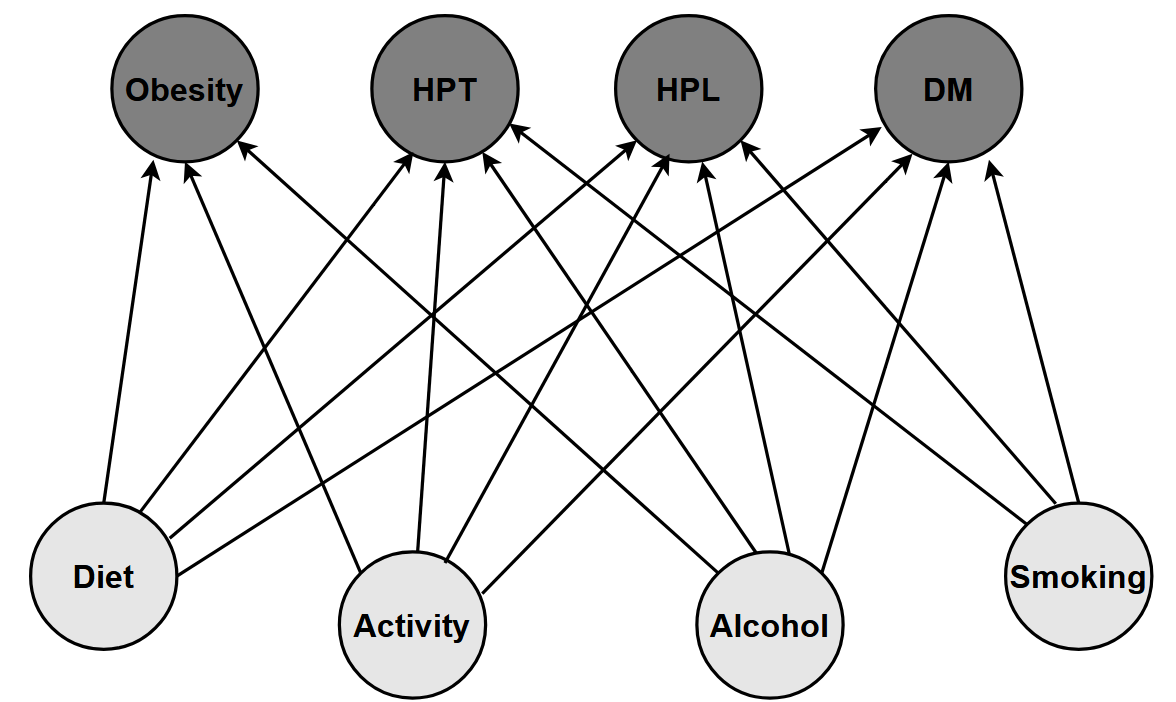}
    \caption{Causal graph of features which affect the output of interest ``Risk of Cardio-vascular disease (CVD)". There are $8$ features, all of which are \textit{causal}. The features at the bottom form the set of desirable features $\des$ and those on the top form the set of undesirable features $\und$. The causal links are indicated on the graph. This causal graph has a special structure: it is \textit{bipartite}. The edge weights are recorded in Table~\ref{table:weights} in Appendix~\ref{sec:app_addn_exp}.}
    \label{fig:causal_graph}
\end{figure}

\emph{Generating Prior Beliefs} We note that our desirable features are generally harder to observe than our undesirable features. First, DM, HPL, HPT, and Obesity are easy-to-quantify features that are also verifiable by a doctor (e.g., though blood work). On the other hand, lifestyle habits are not only hard to observe, but also often mis-reported to clinicians (i.e., under-reporting alcohol consumption, or lying about smoking to avoid insurance upcharges). Hence, we generate all our priors of $h$ to have both a mean and variance of $0$ for all desirable features (i.e., it is fully known that desirable features are not observed by a clinician, and so not used in the clinician's classifier for high risk of CVD), as they are effectively \emph{unobservable}. We consider four mean beliefs $\mu_h$ on the vector $h$, that we denote as follows: 
\begin{itemize}
\item \emph{DM}: There is a weight of $1$ on the ``DM'' feature, and $0$ on all others.
\item \emph{HPL}: There is a weight of $1$ on the ``HPL'' feature, and $0$ on all others.
\item \emph{HPT}: There is a weight of $1$ on the ``HPT'' feature, and $0$ on all others.
\item \emph{Obesity}: There is a weight of $1$ on the ``Obesity'' feature, and $0$ on all others.
\end{itemize}
We note that all other beliefs that only use undesirable, observable features are a linear combination of the four beliefs above, so our insights extend to general classifiers. Also, we demonstrate later that while the classifier does not put any weight on unobserved/desirable features, agents may still exert effort on them because they affect the observed/undesirable features used in the classifier. 

The variance of each of the desirable features is taken to be $0$ (there is complete information that no weight is put on these features in the principal's classifier). Further, we assume that all undesirable features have the same variance, which is parametrized by $\sigma > 0$ --- thus $\sigma$ is a measure of the level of incomplete information. Finally, the covariance matrix of the classifier ($\Sigma_h$) is taken to be diagonal for simplicity of exposition and interpretation, i.e., individuals' beliefs do not encode correlations between features in the deployed classifier.

\subsection{Experimental Results}
Note that we are operating under the regime where the agent has uncertainty only over the classifier $h$ and not over the causal graph (Model $1$), hence the contribution matrix $\C$ is fully known to agents. Agents are assumed to have unweighted $\ell_2$-norm costs.   

For each of the four classifiers described above, we document the \textit{mean contribution} of each feature, given by $\mu_{\C h} = \C \mu_h$ and the $\ell_2$ norm of the mean contribution over the set of desirable and undesirable features, given by $\ell_2(\des)$ and $\ell_2(\und)$ respectively. All values are recorded in Table~\ref{table:mu}. Also, note that due to our choice of the variance structure on $h$ (given by matrix $\Sigma_h = Diag(0,0,0,0,\sigma^2, \sigma^2, \sigma^2, \sigma^2)$ ), all four classifiers have the same covariance over the contribution vector $\C h$ given by $Cov(\C h) = \C \Sigma_h \C^{\top}$.\\

\begin{table}[t!]
\centering
\begin{tabular}{|c|| c|c|c|c || c|c|c|c||c|c|}
\hline 
Classifier & Alcohol & Diet & Activity & Smoking & DM & HPL & HPT & Obesity & $\ell_2(\des)$ & $\ell_2(\und)$
\\
\hline 
\hline 
DM & 0.1 & 0.84 & 0.82 & 0.52 & 1 & 0 & 0 & 0 & 1.28 & 1 
\\
\hline 
HPL & 0.14 & 0.84 & 0.82 & 0.34 & 0 & 1 & 0 & 0 & 1.23 & 1 
\\
\hline 
HPT & 0.62 & 0.84 & 0.82 & 0.86 & 0 & 0 & 1 & 0 & 1.58 & 1 
\\
\hline 
Obesity & 0.64 & 0.86 & 0.82 & 0 & 0 & 0 & 0 & 1 & 1.35 & 1 
\\
\hline
\end{tabular}
\caption{Mean contribution vector $\mu_{\C h}$ for the 4 classifiers: DM, HPL, HPT, Obesity}\label{table:mu}
\end{table}

\paragraph{Desirable features can be incentivized even if they are never observed.} In our four classifiers of choice, observe that no weight has been put on any of the features in set $\des$ because they represent features which cannot be directly observed. However, Figures~\ref{fig:beta_vs_sigma} and \ref{fig:l2_unknown_classifier}  demonstrate that agents still choose to invest significant effort into desirable features. This is a direct result of the \emph{causal relationship between features}. Observe that in the causal graph, the desirable features influence multiple undesirable features simultaneously. This means that an agent obtains a larger improvement in the undesirable features (which actually affect the agent's classification), not by modifying them directly, but by investing effort into desirable features.

\paragraph{Effect of total contribution on desirable vs undesirable features:} From table ~\ref{table:mu}, it is clear that all four classifiers have a higher $\ell_2$-norm on mean contribution over the set of desirable features compared to undesirable features. Naturally, we expect all 4 classifiers to achieve $\beta > 0.5$ under most circumstances (except at high levels of uncertainty), i.e., more than half of the total magnitude of effort should be invested into desirable features (as per Corollary~\ref{cor:prop_effort}). Indeed, this is in line with what we observe in Figure~\ref{fig:beta_vs_sigma}. We also note that, as expected, HPT is better than Obesity, which is better than DM, which is better than HPL for incentivizing desirable features: this is consistent with the ordering over $\ell_2(\des)$ across these four classifiers, again following the insights of Corollary~\ref{cor:prop_effort}.     

\paragraph{Effect of uncertainty level $\sigma$.} In Figure~\ref{fig:beta_vs_sigma}, we plot how $\beta$ varies as a function of the uncertainty parameter $\sigma$ for different classifiers. Higher $\sigma$ indicates a higher degree of uncertainty about the classifier. As $\sigma$ increases, all four classifiers degrade in terms of desirability ($\beta$). This is intuitive: at higher levels of uncertainty, the contribution of desirable features sees higher variance, as they affect not only themselves but also other features. Undesirable features then become safer to modify.

\begin{figure}[!ht]
  \centering
  \raisebox{20pt}{\parbox[b]{.11\textwidth}{}}%
  \subfloat[][$\delta = 0.1$]{\includegraphics[width=.33\textwidth]{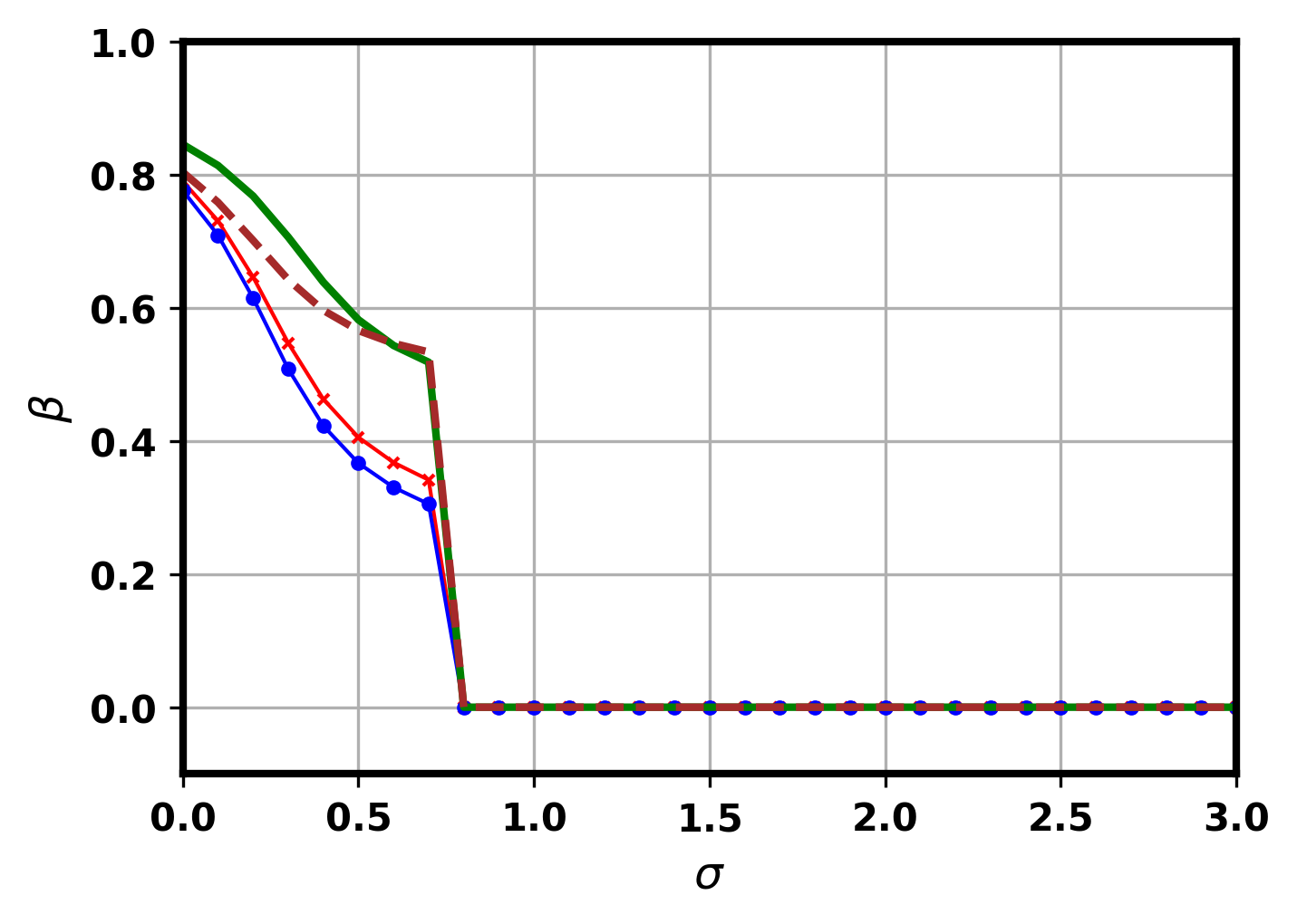}}\hfill
  \subfloat[][$\delta = 0.3$]{\includegraphics[width=.33\textwidth]{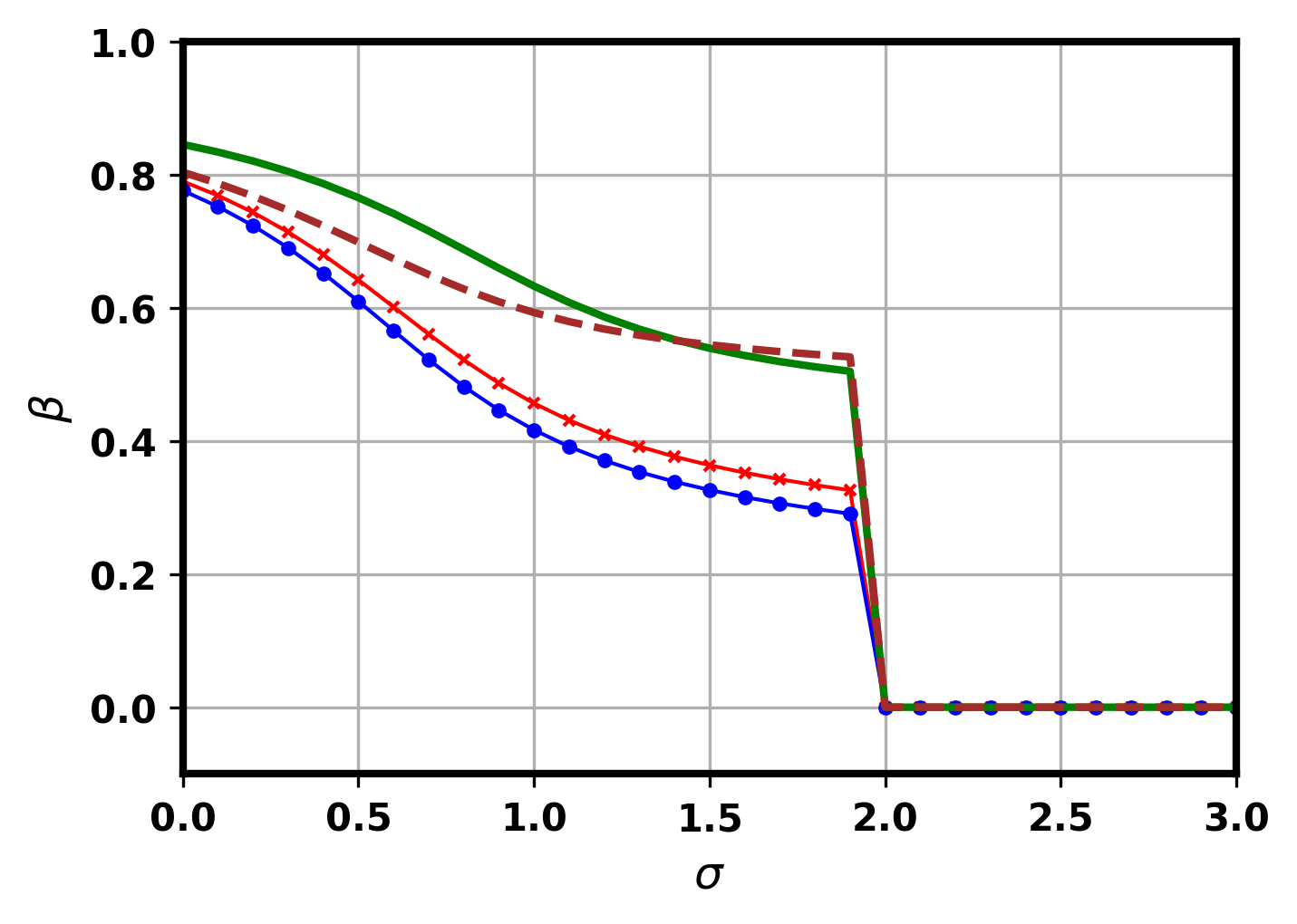}}\hfill
  \subfloat[][$\delta = 0.5$]{\includegraphics[width=.33\textwidth]{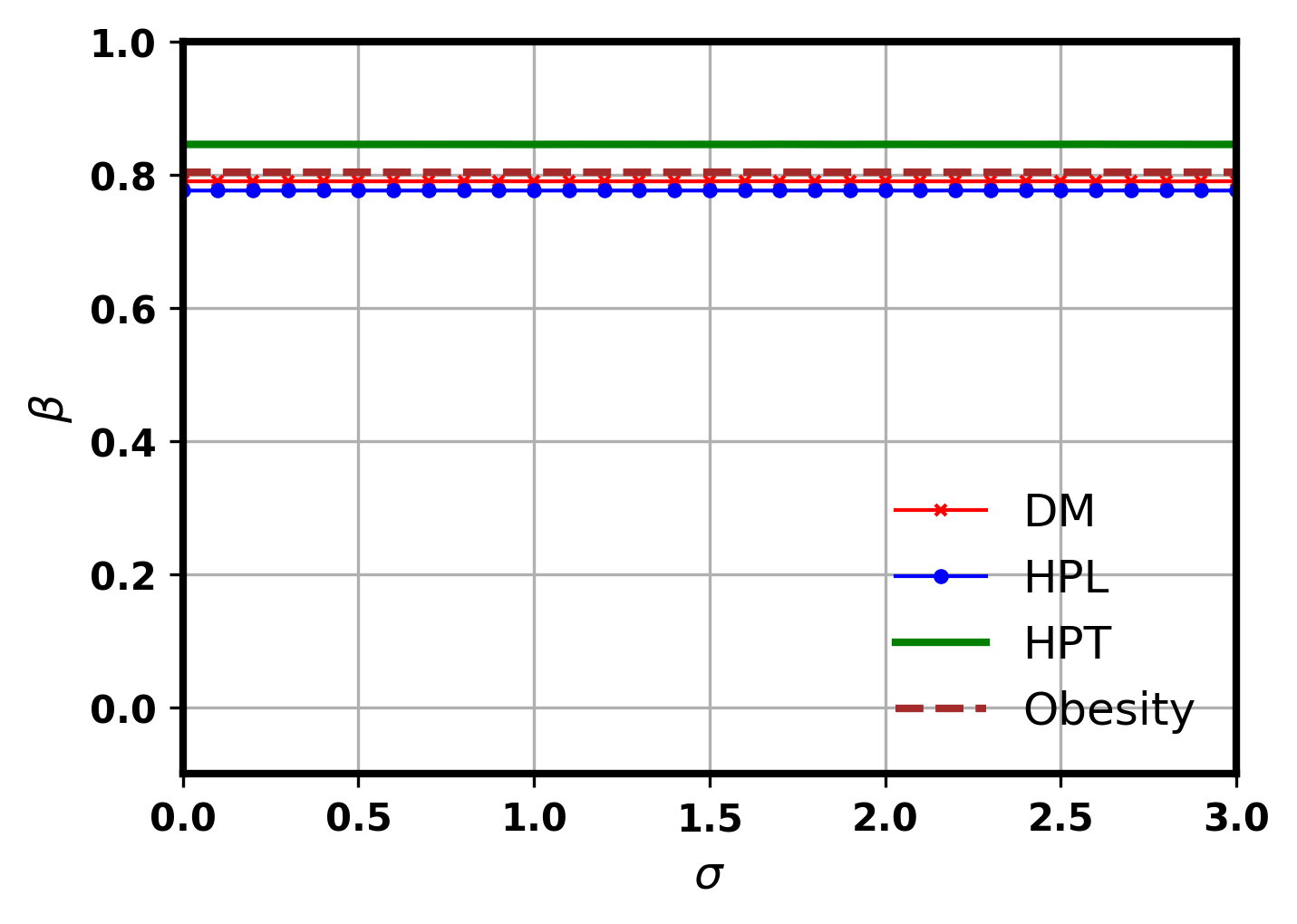}}\par
  \caption{Plot of how $\beta$ varies with $\sigma$ at different levels of $\delta$ and for different classifiers. 
  } 
  \label{fig:beta_vs_sigma}
\end{figure}

\paragraph{Effect of the failure probability $\delta$.} At a fixed level of uncertainty $\sigma$, as the failure probability $\delta$ increases, $\beta$ improves across all four classifiers (Figure~\ref{fig:l2_unknown_classifier}). This is again expected---higher $\delta$ means that the agent is less stringent on the coverage probability requirement and therefore has a much larger space of feasible effort profiles to choose from. Since all features have equal costs, her best response is to invest more in desirable features because they have a higher net contribution which means that she can now ``pass" the classifier while incurring a lower cost.

\paragraph{Trade-offs between $\sigma$ and $\delta$.} The failure probability $\delta$ is closely related to the level of uncertainty $\sigma$. At a fixed level of uncertainty $\sigma$, there is a limit on how low a failure probability $\delta$ can be achieved (Figure~\ref{fig:l2_unknown_classifier}). Similarly, in order to achieve a given failure rate $\delta$, there is a maximum amount of uncertainty $\sigma$ that can be tolerated (Figure~\ref{fig:beta_vs_sigma}); beyond that the problem becomes infeasible. This closely tracks our theoretical findings in Section~\ref{sec:incomplete} (Lemma~\ref{lem:delta_charac}). As $\delta$ increases (the agent imposes a weaker requirement on the coverage probability), a higher degree of uncertainty can be tolerated. Finally at $\delta = 0.5$, the level of uncertainty becomes irrelevant --- this is because at $\delta = 0.5$, the agent wants to pass the classifier with a probability of $\frac{1}{2}$ and therefore can afford to make decisions just on the basis on the mean belief classifier $\mu_h$. 

\begin{figure}[!ht]
  \centering
  \raisebox{20pt}{\parbox[b]{.11\textwidth}{}}%
  \subfloat[][$(\alpha = 1, \sigma = 3)$]{\includegraphics[width=.3\textwidth]{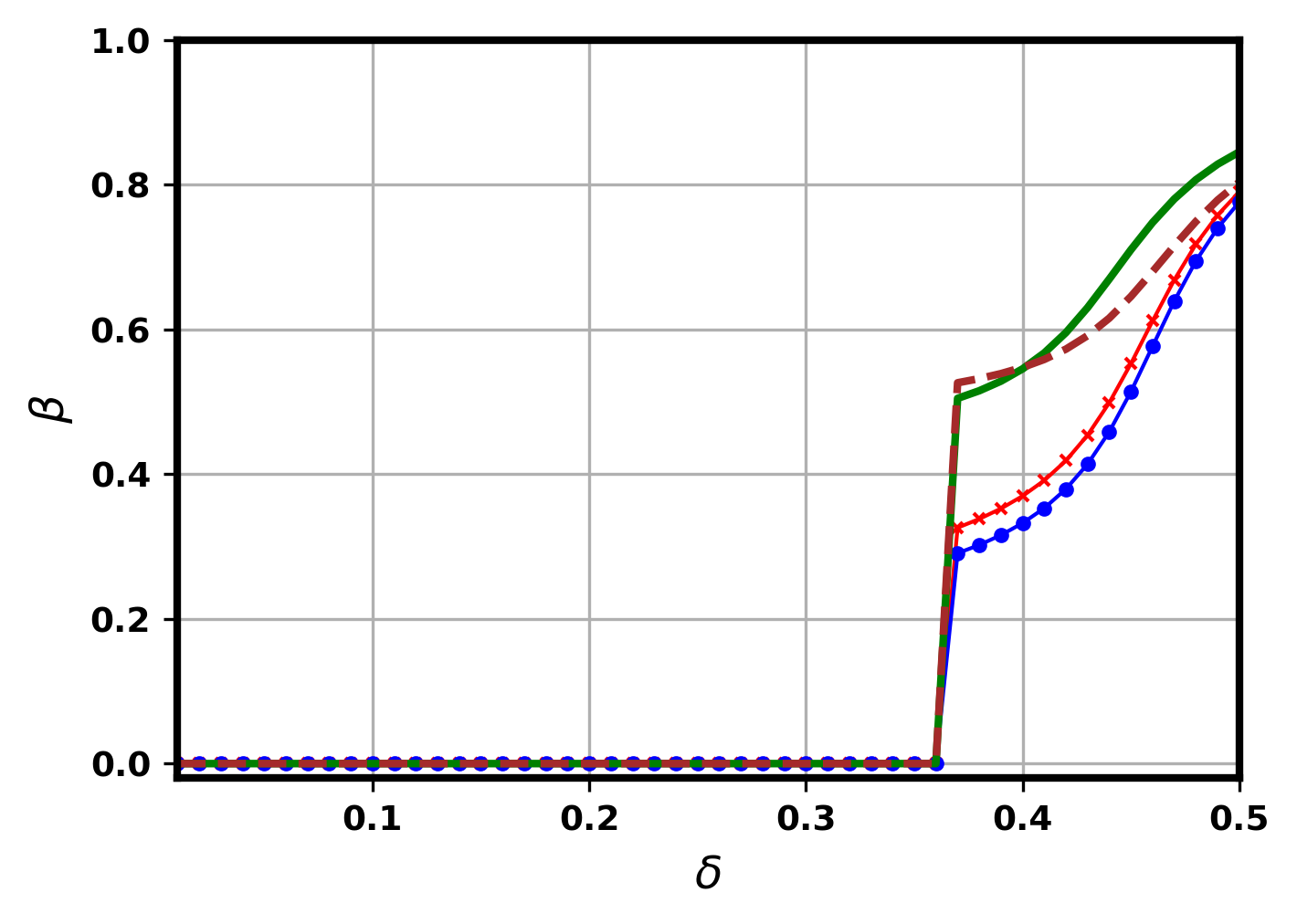}}\hfill
  \subfloat[][$(\alpha = 1, \sigma = 1)$]{\includegraphics[width=.3\textwidth]{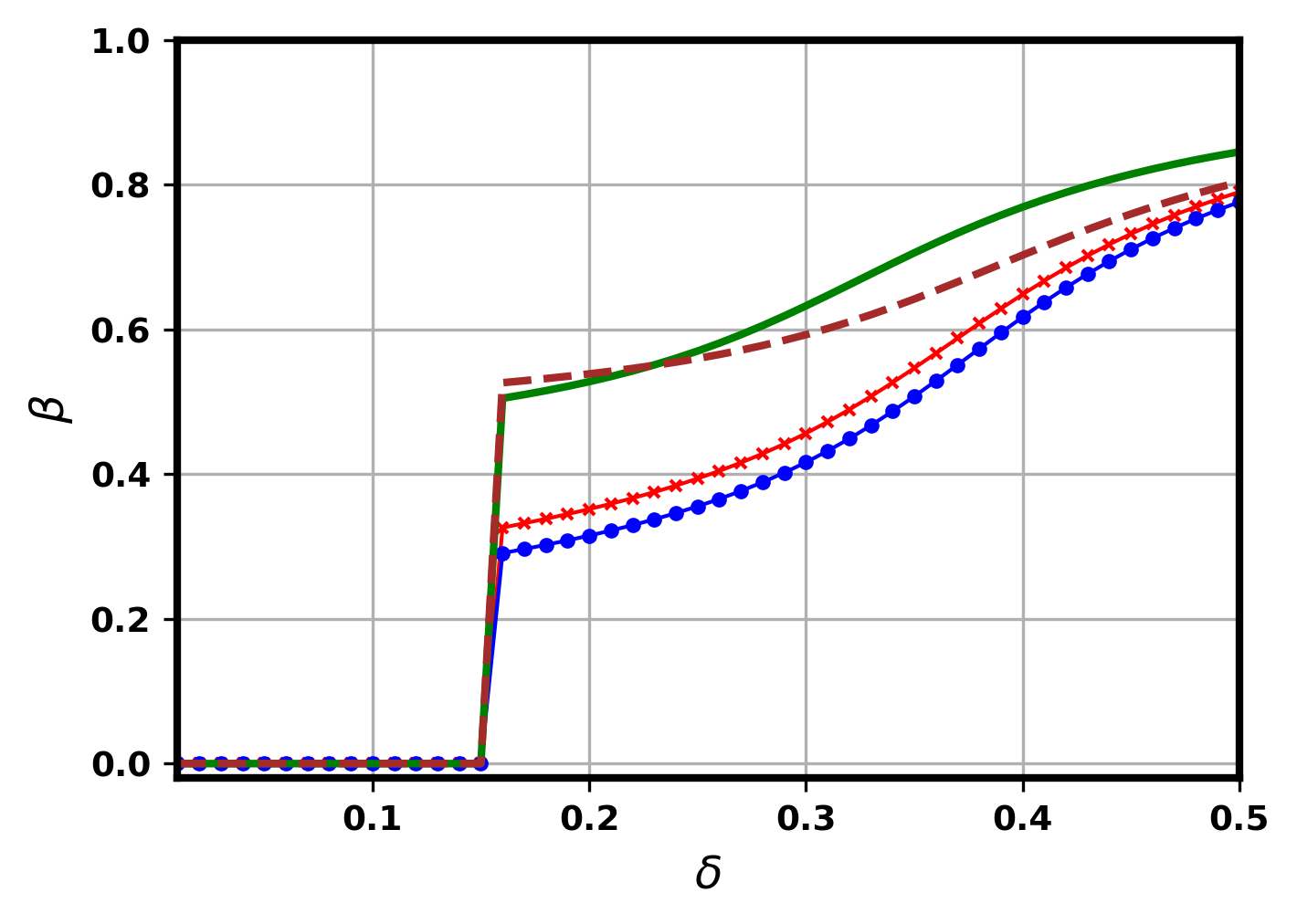}}\hfill
  \subfloat[][$(\alpha = 1, \sigma = 0.1)$]{\includegraphics[width=.3\textwidth]{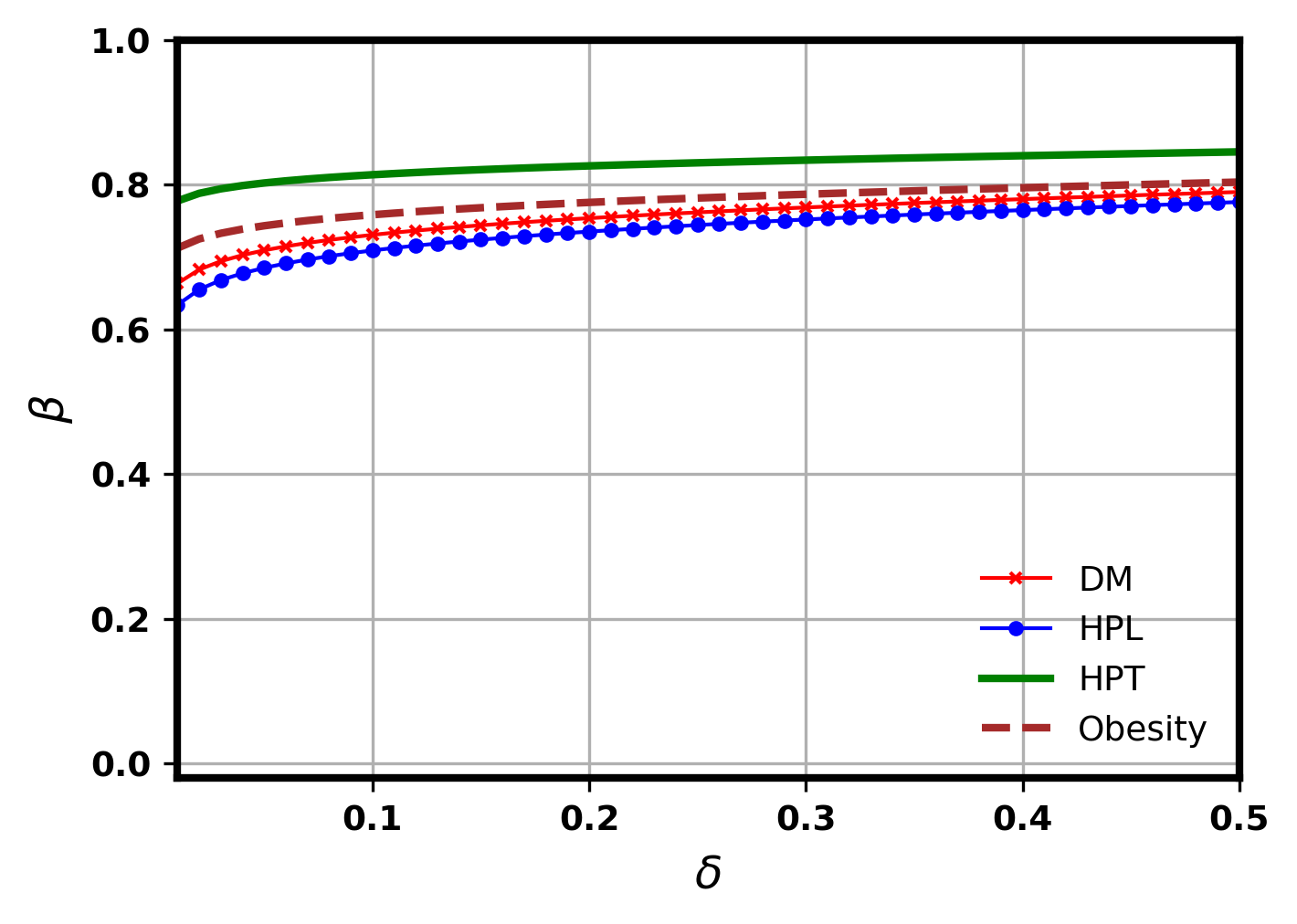}}\par
  \raisebox{20pt}{\parbox[b]{.11\textwidth}{}}%
  \subfloat[][$(\alpha = 10, \sigma = 3)$]{\includegraphics[width=.3\textwidth]{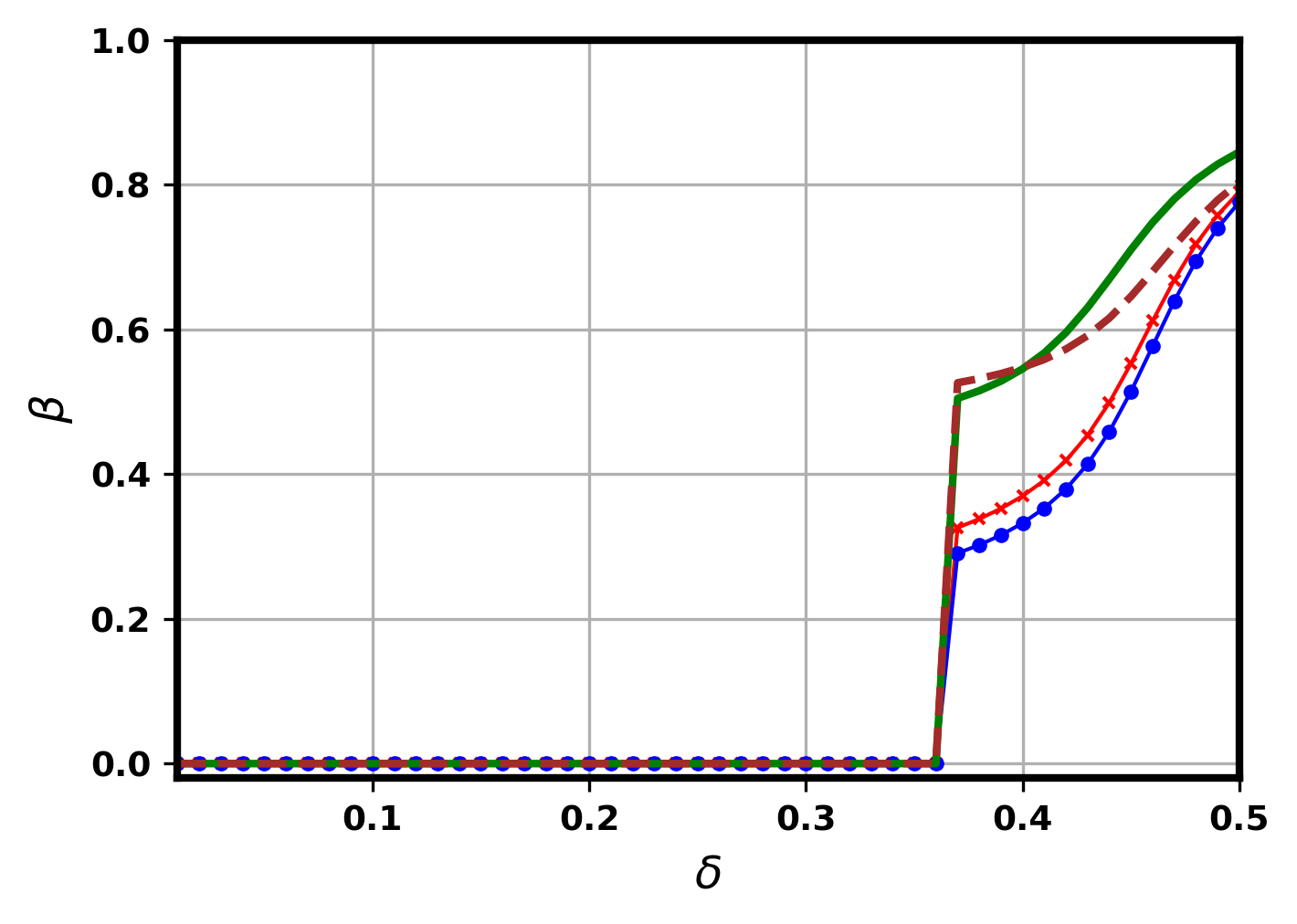}}\hfill
  \subfloat[][$(\alpha = 10, \sigma = 1)$]{\includegraphics[width=.3\textwidth]{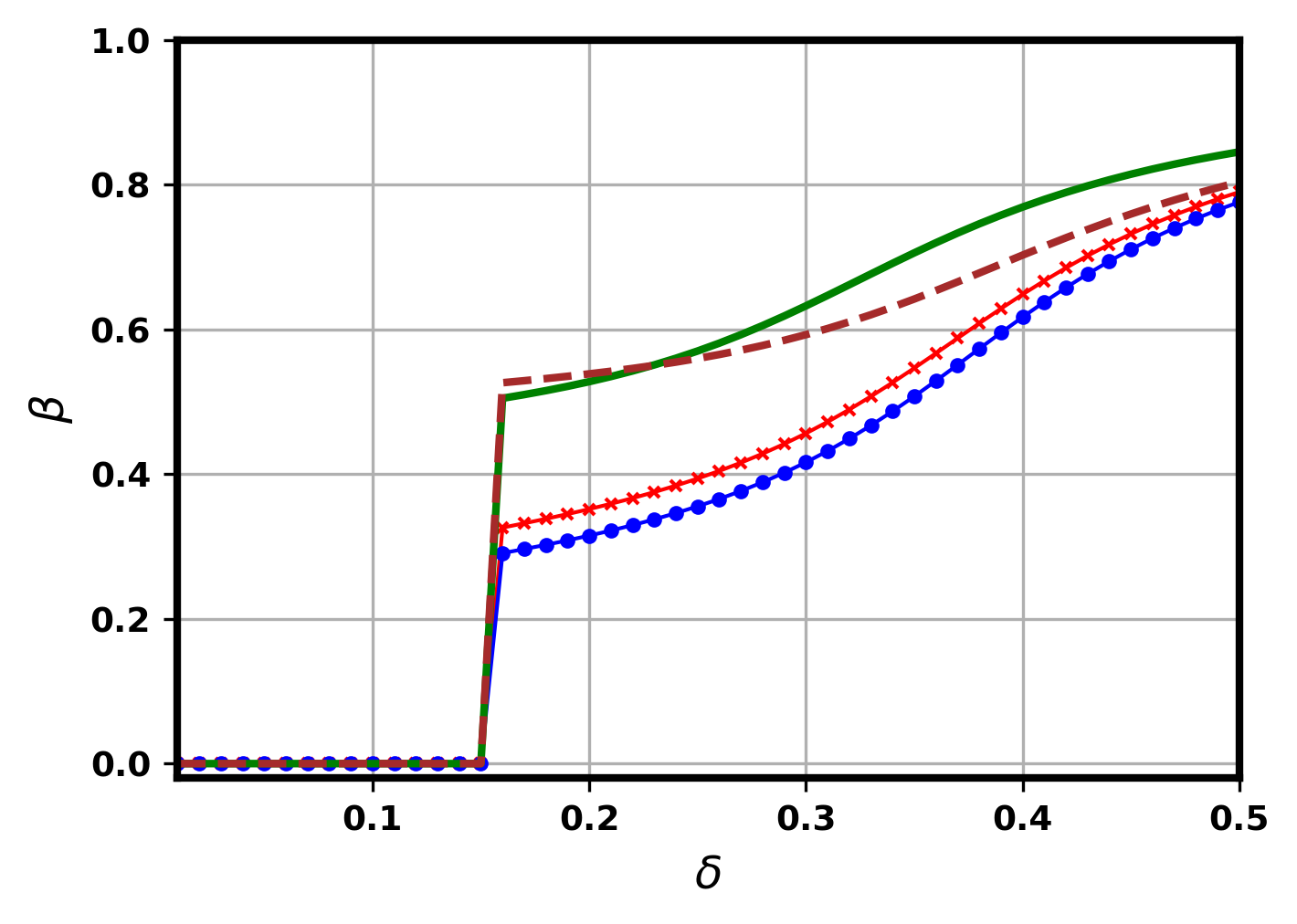}}\hfill
  \subfloat[][$(\alpha = 10, \sigma = 0.1)$]{\includegraphics[width=.3\textwidth]{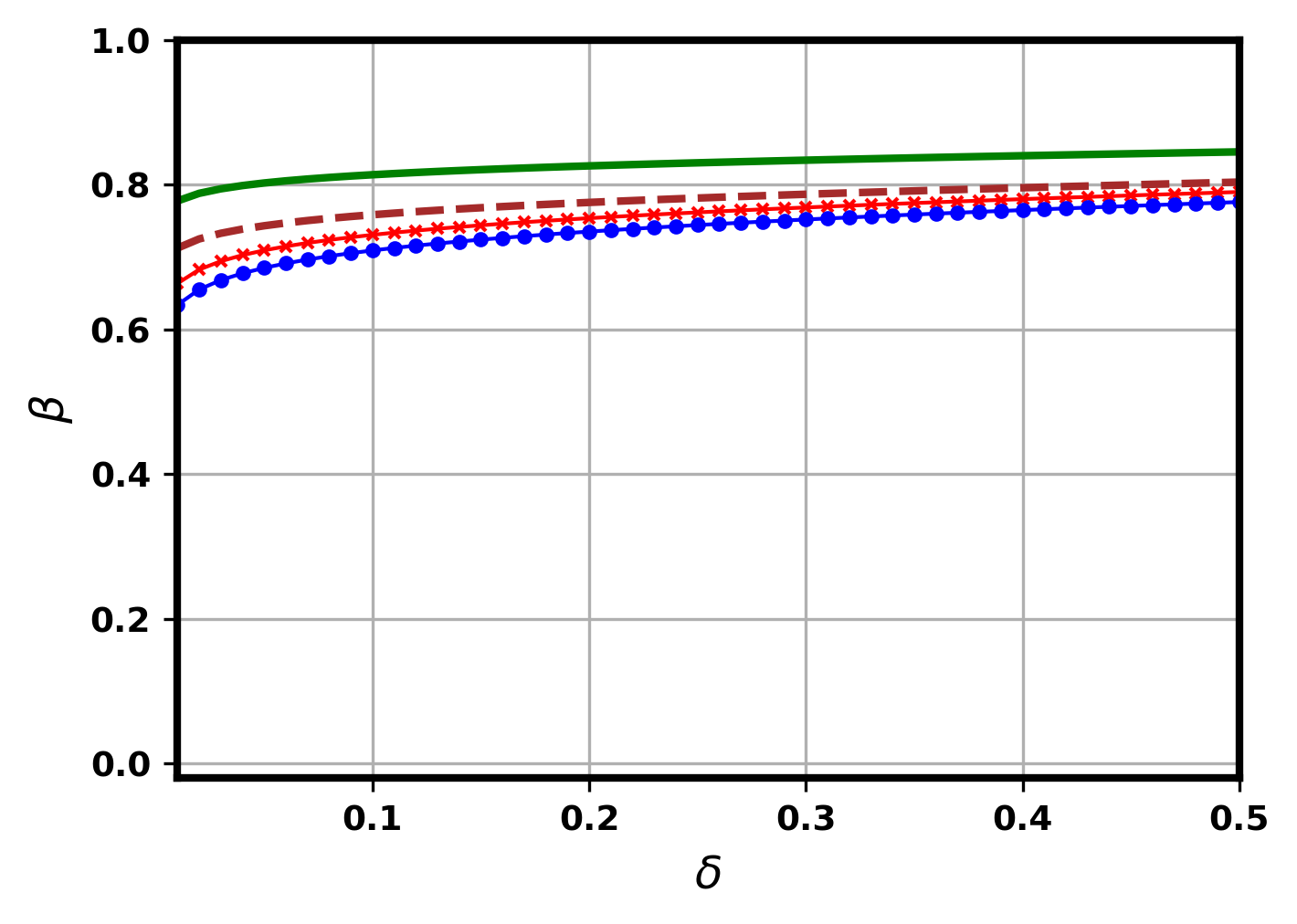}}\par
  \caption{Plots of how $\beta$ varies with $\delta$ for different parameter combinations.} 
  \label{fig:l2_unknown_classifier}
\end{figure}

\paragraph{Effect of $\alpha$.} $\alpha$ represents the amount by which the agent is shy of a positive classification outcome. In this case, we see that $\beta$ has no dependence on $\alpha$ (Figure~\ref{fig:l2_unknown_classifier}). This is an artifact of the $\ell_2$-norm cost function --- the agent's optimal effort profile is proportional to $\alpha$ in each feature and therefore $\beta$ remains unaffected. However, note that $\alpha$ does affect the cost incurred by the agent, the farther she is from the classification boundary (higher $\alpha$), the higher is the cost incurred.

\section{Discussion}\label{sec:discussion}
In this paper, we adopt a causal perspective to the problem of strategic classification. The principal deploys a linear classifier which classifies agents ``positive" or ``negative" based on a set of features embedded on a \textit{causal graph}. Since agents are strategic, they are expected to invest effort ``cleverly" to modify their features in the hopes of a positive classification outcome while incurring the minimum possible cost. 
Therefore, understanding how agents respond when they have different levels of access to information about the deployed classifier and the causal graph, is significant to the principal from the perspective of classifier design. The principal's goal is the following: design a classifier which incentivizes agents to invest in \textit{desirable} effort profiles. 

The main contributions of our paper are two-fold: i) We characterize the design space of \textit{desirable} classifiers for the broad class of weighted $\ell_p$-norm agent cost functions under complete information, while clearly demonstrating computational challenges of finding such classifiers. We also identify special settings and relaxations which can render the design problem computationally tractable. ii) We try to understand strategic agent behavior in response to classifiers under incomplete information settings. In particular, we show that under totally incomplete information (uncertainty over both the classifier and the causal graph), finding the agent's best response is computationally difficult. However it becomes tractable under partially incomplete information where there is uncertainty over either the classifier or the causal graph and we provide insights about the structural properties of the agent's best response under this setting. Finally, we use these results to gain some useful insights (through numerical experiments) into the question of how to design \textit{desirable} classifiers, even under incomplete information. 

There are many avenues of future work. Our model of uncertainty involves agents having Gaussian priors over the classifier or the edge weights of the causal graph or both, under the assumption that the causal graph topology is always fully known. In real life, there may be other forms of uncertainty --- for example, when there are a large number of features, it may be unreasonable to assume that agents have complete knowledge about the causal relationships between features. It may also be interesting to explore if there are other kinds of information structures which are more interpretable --- for example, instead of priors independently on the classifier and the causal graph, agents have access to an ordering (or partial ordering) on features in terms of their relative importance. This information structure subsumes necessary information from both the classifier and the causal graph but is easier to understand and hence, might be easier to respond to. Our work also has interesting extensions in the domain of fairness. Different populations may have different levels of uncertainty in their priors which might lead to markedly different abilities of each of those groups to respond to the principal's classifier --- the downstream fairness in classification of such information asymmetry may be worth exploring.

\bibliographystyle{plainnat}
\bibliography{mybib.bib}

\newpage
\appendix
\newpage
\section{Additional Details for Experimental Section~\ref{sec:experiments}}\label{sec:app_addn_exp}
\subsection*{Supplementary Tables}
\begin{table}[h!]
\centering
\begin{tabular}{|c|| c|c|c|c || c|c|c|c||}
\hline 
Features & Alcohol & Diet & Activity & Smoking & DM & HPL & HPT & Obesity
\\
\hline 
\hline 
Alcohol & 0 & 0 & 0 & 0 & 0.10 & 0.14 & 0.62 & 0.64 \\
\hline 
Diet & 0 & 0 & 0 & 0 & 0.84 & 0.84 & 0.84 & 0.86 \\
\hline 
Activity & 0 & 0 & 0 & 0 & 0.82 & 0.82 & 0.82 & 0.82 \\
\hline 
Smoking & 0 & 0 & 0 & 0 & 0.52 & 0.34 & 0.86 & 0 \\
\hline
DM & 0 & 0 & 0 & 0 & 0 & 0 & 0 & 0 \\
\hline
HPL & 0 & 0 & 0 & 0 & 0 & 0 & 0 & 0 \\
\hline
HPT & 0 & 0 & 0 & 0 & 0 & 0 & 0 & 0 \\
\hline
Obesity & 0 & 0 & 0 & 0 & 0 & 0 & 0 & 0 \\
\hline
\end{tabular}
\caption{Adjacency matrix $A$ which captures the edge weights of the causal graph in Figure~\ref{fig:causal_graph}}\label{table:weights}
\end{table}

\section{Proofs for the Complete Information Setting}

\subsection{Proof of Proposition~\ref{prop:y_pos}}\label{app:y_pos}
Let $\effopt$ be the optimal effort profile for the agent, i.e., the profile that corresponds to the solution of \ref{opt:agent_fullinfo}. First, note that for any feature $f \in \cF$, $(\C h_0)_f = 0$ implies $\effopt_f = 0$. Indeed, if feature $f$ has no contribution towards the classification decision, then no effort should be expended on $f$ in the optimal effort profile. 

Now, we can focus on features for which $(\C h_0)_f \neq 0$. When $(\C h_0)_f < 0$, we will show that $\effopt_f \leq 0$. Suppose that $\effopt_f > 0$. In this case, we can construct a new effort profile $\eff'$ as follows: $\eff'_f = 0$ and $\eff'_g = \effopt_g$ for all $g \in \cF, g \neq f$. It is easy to see that $\eff'$ is still feasible but $\textsf{Cost}(\eff') < \textsf{Cost}(\effopt)$ which contradicts the fact that $\effopt$ is the optimal solution. Therefore, $\effopt_f \leq 0$. Similarly, we can show that when $(\C h_0)_f > 0$, $\effopt_f \geq 0$. 

The above discussion implies that whenever $(\C h_0)_f \neq 0$, then $(\C h_0)_f \effopt_f \geq 0$ - therefore, without loss of generality, it suffices to assume $(\C h_0)_f > 0$ and only search over the space $\eff \geq 0$ (because the optimal solution $\effopt \geq 0$). This concludes the proof.

\subsection{Proof of Lemma~\ref{lem:linear_onefeature}}\label{app:linear_onefeature}

The optimization problem in \eqref{opt:linear} (which we will call the primal problem $P$) is a linear program whose feasible region is given by the following polyhedron $Q = \left\{\eff \in \R^{n+k}: (\C h_0)^{\top}\eff \geq \alpha, \eff \geq 0\right\}$. Our first goal is to argue that the optimal solution is a corner point of $Q$ which requires us to prove the following: i) firstly, $Q$ has at least one corner point, and ii) the optimal solution is bounded which would imply that it must be at a corner point of $Q$. Note that the polyhedron $Q$ has no line (because it is a subset of the positive orthant) and therefore, it must have at least one corner point\footnote{for more details on the polyhedral theory related to linear optimization, refer to \cite{bertsimas}}. We can now write the dual problem ($D$) as follows:
\begin{align*}
    \max_{\pi} \quad &\alpha \pi \quad s.t. \quad \quad (D)\\
    & \pi \cdot (\C h_0) \leq c    \quad (\eff)\\
    & \pi \geq 0. 
\end{align*}
We know that the dual problem ($D$) is feasible ($\pi = 0$ is feasible to $D$) which implies that the optimal solution to ($P$) cannot be unbounded. Hence, we conclude that there must exist a corner point optimal solution to problem ($P$). Now, note that all corner points of $Q$ are obtained by the intersection of the hyperplane $(\C h_0)^{\top}\eff \geq \alpha$ with the positive axes. So any corner point of $Q$ must be of the form where exactly one entry corresponding to some feature $f$ is positive (i.e., takes value $\frac{\alpha}{(\C h_0)_f}$) and all other entries are zero. This implies that there exists an optimal effort profile where the agent needs to modify exactly one feature, proving the first part of the lemma. 

For proving the second part, we will use the complementary slackness conditions on the dual constraints. We already know that there exists an optimal primal solution where there is some feature $\fstar$ with $\eff_{\fstar} = \frac{\alpha}{(\C h_0)_{\fstar}}$ and $\eff_{f'} = 0$ for all $f' \neq \fstar$. Let $\pistar$ be the optimal dual solution. Using complementary slackness, we know that $\pistar \cdot (\C h_0)_{\fstar} = c_{\fstar}$ which implies that: 
\[
        \pistar = \frac{c_{\fstar}}{(\C h_0)_{\fstar}}. 
\]
Since $\pistar$ must also be feasible to ($D$), we must have: 
\[
       \pistar \leq \frac{c_{f'}}{(\C h_0)_{f'}} \quad \forall~f' \neq \fstar,~(\C h_0)_{f'} \neq 0,  
\]
which implies that: 
\[
        \fstar \in \arg\max_{f \in \cF} \frac{(\C h_0)_f}{c_f}.  
\]
This concludes the proof of the lemma.

\subsection{Proof of Lemma~\ref{lem:l2_effort}}\label{app:l2_effort}

We solve the constrained optimization problem using Lagrange multipliers. Define the Lagrangian as follows:
\[
    \mathcal{L}(\eff, \pi) = \left(\sum_{f \in \cF}c_f (\eff_f)^p\right)^{1/p} +\pi \left(- (\C h_0)^{\top}\eff + \alpha \right),
\]
where $\pi$ is the Lagrange multiplier associated with the constraint as defined earlier. This gives us the following set of KKT conditions:
\begin{align*}
    &\nabla_{\eff} \mathcal{L}(\eff, \pi) = 0, \\
    &\pi \cdot (-(\C h_0)^{\top}\eff + \alpha ) = 0, \\
    &\pi \geq 0, \\
    &\alpha - (\C h_0)^{\top}\eff \leq 0, ~\eff \geq 0. 
\end{align*}
Since our optimization problem is convex, it suffices to find a pair $\left(\effopt, \pistar \right)$ that satisfies the KKT conditions and we can automatically conclude that $\effopt$ is optimal to the primal problem.

First, we show that the constraint $(\C h_0)^{\top}\eff \geq \alpha$ must be active at the optimal solution. We prove this by contradiction. 
Suppose, if possible that $-(\C h_0)^{\top}\effopt + \alpha < 0$. However, this means that we can obtain the optimal solution $\effopt$ by solving the primal problem as if it were unconstrained. In that case, it must be that $\effopt = 0$, but observe that $\eff = 0$ is not even feasible (and hence cannot be optimal). This implies that the constraint must hold at equality. Therefore, we can solve for $\effopt$ and $\pi^*$ by solving the following system:
\begin{align*}
    -(\C h_0)^{\top}\eff + \alpha &= 0,\\
    \nabla_e \mathcal{L}(\eff, \pi) &= 0.
\end{align*}
Now, 
\[
    \left(\nabla_{\eff} \mathcal{L}(\eff, \pi) \right)_f = \frac{\partial \mathcal{L}}{\partial \eff_f} = \frac{{c_f (\eff_f)}^{p-1}}{ \left[ \left(\sum \limits_{f \in \cF}c_f (\eff_f)^p \right)^{1/p} \right]^{p-1} } - \pi \cdot \left( \C h_0 \right)_f.
\]
We have already argued that $\effopt \neq 0$. Therefore, $\pistar > 0$. This implies that for all features $f \in \cF$, whenever $(\C h_0)_f > 0$, we must have: 
\[
        \effopt_f \propto \left(\frac{(\C h_0)_f}{c_f}\right)^{1/(p-1)}, 
\]
and when $(\C h_0)_f = 0$, the condition holds trivially. This concludes the proof of the lemma.

\subsection{Proof of Lemma~\ref{lem:comp_des_nonconvex}}\label{app:comp_des_nonconvex}

The set of $\beta$-desirable classifiers $\cH$ is given as follows: 
\begin{align*}
    \cH := \left\{h_0 \in \R^{|\cF|}: \effopt(h_0, \C) \text{ is $\beta$-desirable}, \C h_0 \geq 0 \right\}.
\end{align*}
We define the set $\cZ := \left\{ (\C h_0): h_0 \in \cH \right\}$.
Suppose that $\C$ is full row-rank. This implies that $\cH$ is convex if and only if $\cZ$ is convex\footnote{This is a standard result in linear algebra; the proof provided in Appendix~\ref{sec:app_sup} for completeness}. Therefore, in order to complete the proof, it suffices to show that the transformed set $\cZ$ is non-convex in the worst case. We now provide instances of problems where $\cZ$ is non-convex and the agents incur $\ell_p$-norm cost functions with $p = 1$ and $p > 1$. 
Recall from Theorems~\ref{thm:l1_good} and \ref{thm:l2_good} that the set $\cZ$ is given as follows: 
\[
    \cZ =\left\{ z \in \R_{\geq 0}^{|\cF|}: \max_{f \in \und} \frac{z_f}{c_f} < \max_{f \in \des}\frac{z_f}{c_f} \right\} \quad \text{($p = 1$)}
\]
\[
    \cZ = \left\{ z \in \R_{\geq 0}^{|\cF|}:  \left[ \sum_{f \in \des} \left( \frac{z_f}{c_f} \right)^{2/(p-1)} \right]^{1/2} \geq \frac{\beta}{\sqrt{1-\beta^2}}  \left[ \sum_{f \in \und} \left( \frac{z_f}{c_f} \right)^{2/(p-1)} \right]^{1/2}  \right\} \quad \text{($p > 1$)}
\]

\paragraph{Weighted $\ell_1$-norm cost function:} Consider a setting where there are $4$ features with $\des = \left\{1,2\right\}$ and $\und = \left\{3,4\right\}$. Suppose that the cost vector equals $c = \bf{1}$. In this case, 
\[
       \cZ = \left\{z \in \R_{\geq 0}^4: \max(z_3, z_4) < \max(z_1, z_2)\right\}.
\]
Now, choose $z' :=(4, 7, 3, 6)$ and $z'':= (7, 4, 3, 6)$. Both are clearly points in $\cZ$. However, for $\alpha = 0.5$, $\alpha z' + (1-\alpha)z'':= (5.5, 5.5, 3, 6) \notin \cZ$. Therefore, $\cZ$ is not a convex set.  

\paragraph{Weighted $\ell_p$-norm cost function:} Consider a setting where there are $3$ features with $\des = \left\{1,2\right\}$ and $\und = \left\{3\right\}$. Let $p = 2$, $c = \bf{1}$ and $\beta = \frac{1}{\sqrt{2}}$. Then $\cZ$ is given by:
\[
    \cZ = \left\{z \in \R_{\geq 0}^3: \sqrt{z_1^2 + z_2^2} \geq z_3 \right\}
\]
$(0, 1, 1)$ and $(1, 0, 1)$ are points in $\cZ$, but the point halfway between them, given by $(0.5, 0.5, 1)$ is clearly not in $\cZ$. Therefore, $\cZ$ is not a convex set. This concludes the proof.  

\subsection{Proof of Proposition~\ref{prop:convex_p13}}\label{app:convex_p13}

We will verify convexity separately for the cases with $\ell_1$-norm and $\ell_p$-norm ($1 < p \leq 3$) cost functions. \smallskip

\noindent\emph{The $\ell_1$-norm case:} When the cost function is a weighted $\ell_1$-norm, the set of desirable classifiers is given by 
\[
     \cH := \left\{h_0 \in \R^{|\cF|}: \C h_0 \geq 0, \max_{f \in \und}\frac{(\C h_0)_f}{c_f} < \max_{f \in \des}\frac{(\C h_0)_f}{c_f} \right\}.
\]
Now, suppose $|\des| = 1$ and there is some feature $f_d \in \des$. 
Then, we can rewrite $\cH$ as follows: 
\[
     \cH := \left\{h_0 \in \R^{|\cF|}: \C h_0 \geq 0, \max_{f \in \cF \setminus \left\{f_d\right\}} \frac{(\C h_0)_f}{c_f} - \frac{(\C h_0)_{f_d}}{c_{f_d}} < 0  \right\}. 
\]
In order to show that $\cH$ is a convex set, it suffices to show that the function $g(h_0) = \max_{f \in \cF \setminus \left\{f_d\right\}} \frac{(\C h_0)_f}{c_f} - \frac{(\C h_0)_{f_d}}{c_{f_d}}$ is a convex function. Function $g(\cdot)$ corresponds to the sum of a maximum of linear functions (which is convex) and a linear function; hence, function $g(\cdot)$ is convex.

\smallskip
\noindent\emph{The $\ell_p$-norm case with $p > 1$:} For $\ell_p$-norm cost functions with $p > 1$, set $\cH$ is given by:
\[
    \cH \triangleq \left\{ h_0 \in \R^{|\cF|}: \C h_0 \geq 0,  \left[ \sum_{f \in \des} \left( \frac{(\C h_0)_f}{c_f} \right)^{2/(p-1)} \right]^{1/2} \geq \frac{\beta}{\sqrt{1-\beta^2}}  \left[ \sum_{f \in \und} \left( \frac{(\C h_0)_f}{c_f} \right)^{2/(p-1)} \right]^{1/2}  \right\}
\]
Using the fact that $|\des| = 1$, we can rewrite $\cH$ as follows: 
\[
      \cH := \left\{h_0 \in \R^{|\cF|}: \C h_0 \geq 0, (\C h_0)_{f_d} \geq K \left[ \sum_{f \in \und} \left( \frac{(\C h_0)_f}{c_f} \right)^{2/(p-1)} \right]^{(p-1)/2}  \right\}
\]
where $K = c_{f_d} \left(\frac{\beta}{\sqrt{1-\beta^2}} \right)^{(p-1)} > 0$. 
Now in order to complete the proof, we need to show that the function $r(h_0)$ is convex, where $r(h_0)$ is given by: 
\[
    r(h_0) = K \left[ \sum_{f \in \und} \left( \frac{(\C h_0)_f}{c_f} \right)^{2/(p-1)} \right]^{(p-1)/2} - (\C h_0)_{f_d}.
\]
When $1 < p \leq 3$, we can rewrite $r(h_0)$ as follows: 
\[
    r(h_0) = K \|Bh_0 \|_{2/(p-1)} - (\C h_0)_{f_d},
\]
where $B \in \R^{(|\cF|-1) \times (|\cF|-1)}$. Note that $K \| Bh_0 \|_{2/(p-1)}$ is a convex function in $h_0$ since this is a $q$-norm for $q = \frac{2}{p-1} \geq 1$. This makes $r(h_0)$ a convex function in $h_0$ (sum of a convex function and a linear function is convex) and concludes the proof.

\section{Proofs for the Incomplete Information Setting}

\subsection{Proof of Lemma~\ref{lem:partial_incomp_convex}}
\label{app:partial_incomp_convex}

Since the cost functions defined in Eq.~\eqref{eq:cost} are convex, in order to complete the proof, it suffices to show that the feasible space of the optimization problem in~\eqref{opt:incomp_gaussian}, is convex. We have already argued that under incomplete information models $(1)$ and $(2)$, $\C h$ is a multivariate Gaussian, i.e, $\C h \sim \mathcal{N}(\mu_{\C h}, \Sigma_{\C h})$ for some $\mu_{\C h} \in \R^{|\cF|}$ and $\Sigma_{\C h} \in \R^{|\cF|\times |\cF|}$. This implies, 
\[
       (\C h)^{\top}\eff \sim \mathcal{N}\left(\mu_{\C h}^{\top}\eff, \eff^{\top}\Sigma_{\C h} \eff \right). 
\]
This allows us to rewrite the LHS of the probability constraint as follows: 
\begin{align*}
    \bP\left[(\C h)^{\top}\eff \geq \alpha \right] &= \bP\left[ \frac{(\C h)^{\top}\eff - \mu_{\C h}^{\top}\eff}{\sqrt{\eff^{\top}\Sigma_{\C h} \eff}} \geq \frac{\alpha - \mu_{\C h}^{\top}\eff}{\sqrt{\eff^{\top}\Sigma_{\C h} \eff}} \right] \\
    &= \bP\left[Z \geq \frac{\alpha - \mu_{\C h}^{\top}\eff}{\sqrt{\eff^{\top}\Sigma_{\C h} \eff}} \right] \quad \text{(where $Z \sim \mathcal{N}(0, 1)$)} \\
    &= \Phi^c \left( \frac{\alpha - \mu_{\C h}^{\top}\eff}{\sqrt{\eff^{\top}\Sigma_{\C h} \eff}} \right). 
\end{align*}
Therefore, 
\begin{align*}
    \bP\left[(\C h)^{\top}\eff \geq \alpha \right] \geq 1-\delta &\iff \Phi^c \left( \frac{\alpha - \mu_{\C h}^{\top}\eff}{\sqrt{\eff^{\top}\Sigma_{\C h} \eff}} \right) \geq 1-\delta \\
    &\iff \Phi\left( \frac{\alpha - \mu_{\C h}^{\top}\eff}{\sqrt{\eff^{\top}\Sigma_{\C h} \eff}} \right) \leq \delta \\
    &\iff \frac{\alpha - \mu_{\C h}^{\top}\eff}{\sqrt{\eff^{\top}\Sigma_{\C h} \eff}} \leq p_{\delta}   \quad \text{(where $p_{\delta} = \Phi^{-1}(\delta)$)} \\
    &\iff \alpha - \mu_{\C h}^{\top}\eff - p_{\delta}\cdot \sqrt{\eff^{\top}\Sigma_{\C h} \eff} \leq 0.
\end{align*}
When $\delta = \frac{1}{2}$, $p_{\delta} = 0$ and the above constraint reduces to a polyhedral constraint, making the problem trivially convex. On the other hand, note that when $\delta < \frac{1}{2}$, $p_{\delta} < 0$. Now, since $\Sigma_{\C h}$ is a covariance matrix, it is always symmetric and positive semidefinite and therefore, $\Sigma_{\C h}^{1/2}$ exists (it is also symmetric and positive semidefinite!). In that case, we can express $\sqrt{\eff^{\top}\Sigma_{\C h} \eff}$ as follows: 
\[
    \sqrt{\eff^{\top}\Sigma_{\C h} \eff} = \sqrt{ \eff^{\top}\Sigma_{\C h}^{1/2}\Sigma_{\C h}^{1/2}\eff } = \sqrt{(\Sigma_{\C h}^{1/2}\eff)^{\top}(\Sigma_{\C h}^{1/2}\eff)} = \sqrt{||\Sigma_{\C h}^{1/2}\eff ||_2^2} = ||\Sigma_{\C h}^{1/2}\eff||_2. 
\]
Now, $||\Sigma^{1/2}\eff||_2$ is a convex function in $\eff$ (because all norms are convex functions). Similarly, $-p_{\delta}\cdot ||\Sigma_{\C h}^{1/2}\eff||_2$ is also a convex function because $-p_{\delta} > 0$. The term $\alpha -\mu_{\C h}^{\top}\eff$ is affine in $\eff$ and therefore, convex by default. Putting everything together, we conclude that $\alpha - \mu_{\C h}^{\top}\eff - p_{\delta}\cdot \sqrt{\eff^{\top}\Sigma_{\C h} \eff}$ is a convex function in $\eff$ which makes the constraint: 
\[
     \alpha - \mu_{\C h}^{\top}\eff - p_{\delta}\cdot \sqrt{\eff^{\top}\Sigma_{\C h} \eff} \leq 0
\]
a convex constraint. This concludes the proof of the proposition. 

\subsection{Proof of Proposition~\ref{prop:full_uncert_nonconvex}}\label{app:full_uncert_nonconvex}

In order to complete this proof, it suffices to provide a counter-example where the program in \eqref{opt:agent_br_incomp} is non-convex. Consider the simplest possible setting where there can be uncertainty in both the classifier and the causal graph. Suppose, there is only one feature, i.e., $|\cF| = 1$. Let $\omega \sim \mathcal{N}(0, 1)$ be the random variable that captures the uncertainty in the contribution of the feature (encodes uncertainty in the causal graph) and $h \sim \mathcal{N}(0, 1)$ be the random variable that captures the uncertainty in the classifier weight on the feature. Note that $\omega \perp h$. We will show that the feasible space given by:
\[
       \bP\left[ (\omega h) \eff \geq \alpha \right] \geq 1-\delta
\]
is non-convex, which is equivalent to showing that the function $f(\eff)$ given by:
\[
     f(\eff) = \bP\left[ (\omega h) \eff \geq \alpha \right]
\]
is not concave. Below in Figure~\ref{fig:non_convex_counter}, we plot $f(\eff)$ as a function of one-dimensional effort $\eff$. Since there is no closed-form expression for the distribution of the product of two independent standard normal random variables, we obtain empirical estimates for the probability at each $\eff$ using Monte-Carlo simulations. Clearly, $f(\eff)$ is not concave. 
\begin{figure}[!ht]
    \centering
    \includegraphics[width=0.5\linewidth]{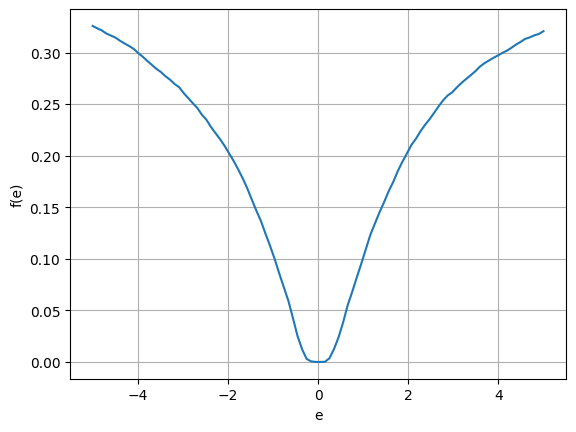}
    \caption{Plot of $f(\eff)$ with $\alpha = 1$}
    \label{fig:non_convex_counter}
\end{figure}
This concludes the proof of the proposition. 

\subsection{Proof of Lemma~\ref{lem:delta_charac}}\label{app:delta_charac}

Recall that the feasible space of the agent's optimization problem in the partially incomplete information case is given by: 
\[
     \alpha - \mu_{\C h}^{\top}\eff - p_{\delta}\cdot \|\Sigma_{\C h}^{1/2}\eff \|_2 \leq 0.
\]
This feasible space is empty at a given $\delta$ if we have: 
\[
    \alpha - \mu_{\C h}^{\top}\eff - p_{\delta}\cdot \|\Sigma_{\C h}^{1/2}\eff \|_2 > 0 \quad \forall~\eff\iff \min_{\eff} \quad g(\eff) = \alpha - \mu_{\C h}^{\top}\eff - p_{\delta}\cdot \|\Sigma_{\C h}^{1/2}\eff \|_2 > 0. 
\]
Now consider the convex unconstrained optimization problem: $\min_{\eff}~g(\eff)$. Then, 
\[
   \nabla g(\eff) = -\mu_{\C h} - p_{\delta}\cdot \frac{\Sigma_{\C h}\eff}{\| \Sigma_{\C h}^{1/2}\eff \|_2}. 
\]
Now there are $2$ cases: 
\noindent\emph{$\nabla g(\eff) = 0$ has a solution $\hat \eff$:} Clearly, $\hat \eff \neq 0$ because the gradient is not defined at $\eff = 0$. In that case, $\hat \eff$ satisfies: 
    \[
          - p_{\delta}\cdot \frac{\Sigma_{\C h}\hat \eff}{\| \Sigma_{\C h}^{1/2}\hat \eff \|_2} = \mu_{\C h}.  
    \]
    Now, $\hat \eff$ must be a global minimizer of $g$ because $g(\cdot)$ is a convex function. We will show that $g(\hat \eff) = \alpha$: 
    \begin{align*}
        g(\hat \eff) &= \alpha - \mu_{\C h}^{\top}\hat \eff - p_{\delta}\cdot \| \Sigma_{\C h}^{1/2}\hat \eff \|_2 \\
        &= \alpha + p_{\delta}\cdot \frac{\hat \eff^{\top}\Sigma_{\C h}\hat \eff }{\| \Sigma_{\C h}^{1/2}\hat \eff \|_2} - p_{\delta}\cdot \| \Sigma_{\C h}^{1/2}\hat \eff \|_2 \quad \text{(using the condition from $\nabla g(\hat \eff) = 0$)}\\
        &= \alpha + p_{\delta}\cdot \| \Sigma_{\C h}^{1/2}\hat \eff \|_2 - p_{\delta}\cdot \| \Sigma_{\C h}^{1/2}\hat \eff \|_2 \\
        &= \alpha. 
    \end{align*}
    Since $\alpha > 0$, the problem is always infeasible for this particular value of $p_{\delta}$. Since $\Sigma_{\C h}$ is positive definite, $\Sigma_{\C h}^{-1/2}$ exists, therefore we have: 
    \[
         p_{\delta} = - \| \Sigma_{\C h}^{-1/2} \mu_{\C h} \|_2 \iff \delta = \Phi^{-1}\left( - \| \Sigma_{\C h}^{-1/2} \mu_{\C h} \|_2 \right). 
    \]

\noindent\emph{$\nabla g(\eff) = 0$ has no solution:} This means that either the unique optimal solution is at the point where the gradient does not exist, i.e, $\eff = 0$, or the solution is unbounded. The first subcase clearly leads to infeasibility as $g(0) = \alpha > 0$ while the second sub-case leads to a non-empty feasible region for Problem \eqref{opt:incomp_gaussian}. We will now try to derive conditions on $\delta$ which lead to each subcase. 

    Suppose that the unique optimal solution is $\eff = 0$. This means that for any direction $d$, $g(0 + d) > g(0)$ or equivalently, 
    \begin{align*}
         \alpha - \mu_{\C h}^{\top}d - p_{\delta}\cdot \| \Sigma_{\C h}^{1/2}d \|_2 > \alpha \quad \forall~d 
         \iff -p_{\delta} \cdot \| \Sigma_{\C h}^{1/2}d \|_2 > \mu_{\C h}^{\top}d \quad \forall~ d.
    \end{align*}
    Note that $\mu_{\C h}^{\top}d = \mu_{\C h}^{\top}\Sigma_{\C h}^{-1/2}\Sigma_{\C h}^{1/2}d = (\Sigma_{\C h}^{-1/2}\mu_{\C h})^{\top}(\Sigma_{\C h}^{1/2}d) \leq \| \Sigma_{\C h}^{-1/2}\mu_{\C h} \|_2 \cdot \|\Sigma_{\C h}^{1/2}d \|_2$ where the last inequality follows from the Cauchy-Schwartz inequality. In fact, for $d^* = \Sigma_{\C h}^{-1}\mu_{\C h}$ (which exists since $\Sigma_{\C h}$ is positive definite and hence, invertible), we have equality. But since $-p_{\delta} \cdot \| \Sigma_{\C h}^{1/2}d \|_2 > \mu_{\C h}^{\top}d$ for all directions $d$, it must hold for $d^*$ as well, which implies: 
    \[
         \| \Sigma_{\C h}^{-1/2}\mu_{\C h} \|_2 \cdot \|\Sigma_{\C h}^{1/2}d^* \|_2 < -p_{\delta} \cdot \|\Sigma_{\C h}^{1/2}d^* \|_2,
    \]
    which means that $-p_{\delta} >  \| \Sigma_{\C h}^{-1/2}\mu_{\C h} \|_2$ or equivalently, $\delta < \Phi^{-1}\left( - \| \Sigma_{\C h}^{-1/2} \mu_{\C h} \|_2 \right)$. 

    Similarly, if the solution is unbounded, there must exist a direction $d'$ at $0$ such that:
    \[
        \alpha - \mu_{\C h}^{\top}d' - p_{\delta}\cdot \| \Sigma_{\C h}^{1/2}d' \|_2 < \alpha, 
    \]
    or equivalently, $-p_{\delta} \cdot \| \Sigma_{\C h}^{1/2}d' \|_2 < \mu_{\C h}^{\top}d'$. Using a similar argument as above, we can show that this can happen only when: 
    \[
         \delta > \Phi^{-1}\left( - \| \Sigma_{\C h}^{-1/2} \mu_{\C h} \|_2 \right).
    \]
\noindent
This concludes the proof of the lemma. 

\subsection{Proof of Lemma~\ref{lem:l1_incomp}}\label{app:l1_incomp}
In order to complete the proof, it suffices to construct an instance of the problem where the optimal effort profile is not a corner point. Consider a setting where $|\cF| = 2$, $\alpha > 0$ and $\delta < \frac{1}{2}$.  Suppose, the features are identical in all respects, i.e., $(\mu_{\C h})_1 = (\mu_{\C h})_2 = \bar \mu > 0$, $\Sigma_{\C h} = \begin{bmatrix}\sigma^2 & 0\\0 & \sigma^2\end{bmatrix}$ and $c_1 = c_2 = c$. Additionally, suppose that $\bar \mu > -p_{\delta}\sigma$. We first make the following observations: 
\begin{itemize}
    \item If $\effopt$ is not a corner point, it must be symmetric, i.e., $\effopt_1 = \effopt_2$. 
    \item Since $\bar \mu > 0$ and $\delta < \frac{1}{2}$, it must be that $\effopt \geq 0$ (otherwise, we have infeasibility).  
\end{itemize}
Now, there are only two possible corner point solutions: either of the form $(\eff, 0)$ or $(0, \eff)$. In order for either of them to be optimal, the constraint must be active at that point. Solving, we obtain: 
\[
      \eff = \frac{\alpha}{\bar \mu + p_{\delta}\sigma}; \quad \text{and} \quad \textsf{Cost} = \frac{c\alpha}{\bar \mu + p_{\delta}\sigma}. 
\]
However, we will now construct a non-corner point solution $(\eff', \eff')$ where the constraint is active and which produces a strictly better objective value. Solving, we obtain:
\[
    \eff' = \frac{\alpha}{2\left( \bar \mu + \frac{p_{\delta}}{\sqrt{2}}\cdot \sigma\right)}; \quad \text{and} \quad \textsf{Cost} = \frac{c\alpha}{\left( \bar\mu + \frac{p_{\delta}}{\sqrt{2}}\cdot \sigma\right)}, 
\]
which is strictly smaller than the earlier cost (since $p_{\delta} < 0$). This concludes the proof. 

\subsection{Proof of Theorem~\ref{thm:incomp_l2}}\label{app:incomp_l2}

We will use the KKT conditions to obtain the agent's optimal effort profile $\effopt$.
The Lagrangian $\mathcal{L}(\cdot, \cdot)$ for the above problem is given by:
\begin{align*}
    \mathcal{L}(\eff, \lambda) = ||\eff||_2 + \lambda \left(\alpha - \mu_{\C h}^{\top}\eff - p_{\delta} \cdot ||\Sigma_{\C h}^{1/2}\eff||_2\right), 
\end{align*}
where $\lambda$ is the Lagrange multiplier. 
We can now write the KKT conditions as follows: 
\begin{align*}  
    & \frac{\eff}{||\eff||_2} + \lambda \cdot \left( -p_{\delta}\cdot \frac{\Sigma_{\C h} \eff}{||\Sigma_{\C h}^{1/2}\eff||_2} - \mu_{\C h} \right) = 0,\\
    & p_{\delta} \cdot ||\Sigma_{\C h}^{1/2}\eff||_2 + \mu_{\C h}^{T}\eff = \alpha, \\
    & \lambda > 0. 
\end{align*}
Since we have a convex program, it is sufficient to find a pair $(\effopt, \lambda^*)$ satisfying the KKT conditions and we can immediately conclude that $\effopt$ is an optimal solution to our original problem. 
Using the first equality above, we infer that the optimal effort $\effopt$ must be of the following form: 
\[
        \effopt = \lambda^* \left(k_1 I + k_2 \Sigma_{\C h} \right)^{-1}\mu,
\]
where $k_1 = \frac{1}{||\effopt||_2} > 0$ and $k_2 = \frac{-\lambda^* p_{\delta}}{||\Sigma_{\C h}^{1/2}\effopt||_2} > 0$ (so, the inverse exists). In order to obtain the exact expression for $\effopt$, we need to use the other equality condition and solve simultaneously for $k_1$, $k_2$ and $\lambda^*$. This concludes the proof of the lemma. 

\subsection{Proof of Proposition~\ref{prop:bipartite}}\label{app:bipartite}
Since $\cG$ is a bipartite graph, the set of nodes (in this case, same as set of features) $|\cF|$ can be partitioned into two sets $\cF_{out}$ and $\cF_{in}$ such that $\cF_{in} \cup \cF_{out} = \cF$, $\cF_{in} \cap \cF_{out} = \emptyset$ and all arcs in $\cA$ are directed from $\cF_{out}$ towards $\cF_{in}$. 

Recall that $\Sigma_{\C h}$ is the covariance matrix of $\C h$ where $\C \sim \Pi_{\C}$ and $h \sim \Pi_h$. However, when there is uncertainty only over the edge weights of $\cG$, it is clear that $\Sigma_{\C h} = Cov(\C h_0)$. Therefore, in order to show that $\Sigma_{\C h}$ is a diagonal matrix, it suffices to show that: 
\[
     \forall~f_1, f_2 \in \cF, f_1 \neq f_2, \quad (\C h_0)_{f_1} \perp (\C h_0)_{f_2},
\]
i.e., $(\C h_0)_{f_1}$ and $(\C h_0)_{f_2}$ are independent random variables. Firstly, observe that for any feature $f \in \cF_{in}$, we must have: 
\[
        (\C h_0)_f = 0.
\]
This is because feature $f$ has no outgoing edges (since $f \in \cF_{in}$) and therefore, $\C_{f,.} = \bf{0}^{\top}$ which implies $(\C h_0)_f = \C_{f,.}h_0 = 0$. This automatically implies that the covariance of $(\C h_0)_f$ with any other random variable is also zero. Therefore, we only need to prove that $Cov\left( (\C h_0)_{f_1}, (\C h_0)_{f_2} \right) = 0$ when $f_1, f_2$ both are in $\cF_{out}$. Note that: 
\[
     (\C h_0)_{f_1} = \sum_{f \in \cF} \C_{f_1, f} h_{0,f} \quad \text{and} \quad (\C h_0)_{f_2} = \sum_{f \in \cF} \C_{f_2, f} h_{0,f}.
\]
Therefore, 
\begin{align*}
    Cov\left( (\C h_0)_{f_1}, (\C h_0)_{f_2} \right) &= Cov\left( \sum_{f \in \cF} \C_{f_1, f} h_{0,f}, \sum_{f \in \cF} \C_{f_2, f} h_{0,f} \right) \\
    &= \sum_{f \in \cF} \sum_{f' \in \cF} (h_{0,f}\cdot h_{0,f'}) \cdot Cov(\C_{f_1,f}, \C_{f_2,f'})
\end{align*}
We now argue case by case: 
\begin{itemize}
    \item $f, f' \in \cF_{out}$: In this case, $Cov(\C_{f_1,f}, \C_{f_2,f'}) = 0$ because there can be no edges from either $f_1$ or $f_2$ to $f$ or $f'$ since all of them are nodes in $\cF_{out}$. 
    \item $f \in \cF_{out}, f' \in \cF_{in}$: In this case, $\C_{f_1,f} = 0$ by the same argument as above. Therefore, the covariance must be $0$.
    \item $f' \in \cF_{out}, f \in \cF_{in}$: In this case, $\C_{f_2,f'} = 0$ which makes the covariance $0$.
    \item $f \in \cF_{in}, f' \in \cF_{in}$: Finally, if both $f$ and $f'$ are in $\cF_{in}$, there can be edges from $f_1$ and $f_2$ towards $f$ and $f'$. But those edges are disjoint and therefore, independent which makes the covariance term $0$.
\end{itemize}
This concludes the proof.

\section{Supplementary Proofs}\label{sec:app_sup}

\subsection{Proof of Observation~\ref{obs:comoute_contri}}
The key step to complete the proof is to show that $A_{ij}^k$ captures the influence exerted by feature $i$ on feature $j$ through a directed path on the graph that is exactly $k$ hops long. We will prove by induction. 

\paragraph{Base case ($k = 0$):} When $k = 0$, there exists no directed path from feature $i$ to feature $j$ unless $i = j$. Therefore, all off-diagonal entries are $0$. The only entries appear on the diagonal because feature $i$ affects itself with a unit positive multiplier. This gives us the identity matrix in $|\cF|$ dimensions which is exactly given by $A^0$. 

\paragraph{General case:} Suppose that the induction hypothesis holds for some $k > 1$. We will now show that it also holds for $k+1$. Note that: 
\[
       A_{ij}^{k+1} = \sum_{n=1}^{|\cF|} A_{in}^{k}\cdot A_{nj}.
\]
Since the induction hypothesis is true, $A_{in}^k$ captures the influence exerted by feature $i$ on feature $n$ through a directed path exactly $k$ hops long. $A_{nj}$ represents the direct influence exerted by feature $n$ on feature $j$ (in exactly $1$ hop). Therefore, the product measures the influence of feature $i$ on feature $j$ exerted on a directed path $k+1$ hops long. The sum over all features in $\cF$ captures all such directed paths from $i$ to $j$. Thus, our induction hypothesis is also true for $k+1$. 

Finally, to compute $\C$, we need to sum the influences of directed paths of all lengths starting at node $i$ and ending in node $j$. Since $\cG$ is a directed acyclic graph with $|\cF|$ nodes, the length of the maximum directed path from $i$ to $j$ is at most $|\cF|-1$ hops long or conservatively $|\cF|$ hops long (note that if there are no directed paths of length $k$ from $i$ to $j$, $A_{ij}^k = 0$. So, it does not hurt to be conservative). This leads to the final expression of $\C$: 
\[
      \C = \sum_{k=0}^{|\cF|} A^k.
\]
To conclude the proof, we need to argue about the time complexity of computing $\C$, given matrix $A$. Multiplying $2$ matrices of size $|\cF| \times |\cF|$ takes $O(|\cF|^3)$ time and we need to execute $O(|\cF|)$ such matrix multiplication steps to compute the different powers of $A$. Therefore, the overall time complexity is polynomial in $|\cF|$.   

\subsection{Proof of Supporting Result in Lemma~\ref{lem:comp_des_nonconvex}}
We made the following observation in our proof of Lemma~\ref{lem:comp_des_nonconvex}:
\begin{blob}
Let $X \in \R^n$ and $M \in \R^{n \times n}$. Define set $Y$ as follows: 
\[
        Y := \{y:~\exists~x \in X \quad \text{s.t.}\quad y = Mx\}
\]
When $M$ is full row-rank, set $X$ is convex if and only if set $Y$ is convex.
\end{blob}
\noindent
We provide a formal proof here. We need to show both directions. 

($\implies$) Suppose, set $X$ is convex. We need to show that set $Y$ is convex. Let $y_1, y_2 \in Y$ such that $y_1 \neq y_2$. Pick any $\lambda \in [0,1]$. Then there must exist $x_1, x_2 \in X$ such that $y_1 = Mx_1$ and $y_2 = Mx_2$. Clearly $x_1 \neq x_2$. Since $X$ is a convex set, $\lambda x_1 + (1-\lambda)x_2 \in X$. This implies, 
\begin{align*}
    \lambda y_1 + (1-\lambda)y_2 &= \lambda Mx_1 + (1-\lambda)Mx_2\\
    &= M \left(\lambda x_1 + (1-\lambda)x_2 \right) \in Y. 
\end{align*}

($\impliedby$) For the other direction, we assume that set $Y$ is convex and we need to show that set $X$ is convex. Pick any two elements $x_1, x_2 \in X, x_1 \neq x_2$ and any $\lambda \in [0,1]$. Let $y_1 = Mx_1$ and $y_2 = Mx_2$. Clearly, $y_1, y_2 \in Y$ (by definition). Note that $y_1 \neq y_2$ (otherwise, we would have $Mx_1 = Mx_2$ which implies that $x_1 - x_2 \in \text{Nullspace}(M)$. But $\text{Nullspace}(M) = \emptyset$ as $M$ is full row-rank). Additionally, $\lambda y_1 + (1-\lambda)y_2 \in Y$ since $Y$ is a convex set. This implies, 
\begin{align*}
    \lambda x_1 + (1-\lambda)x_2 &= \lambda M^{-1}y_1 + (1-\lambda)M^{-1}y_2 \quad \text{($M^{-1}$ exists because $rank(M) = n$)}\\
    &= M^{-1} \left( \lambda y_1 + (1-\lambda)y_2 \right) \in X. 
\end{align*}
The last part follows from noting that $\lambda y_1 + (1-\lambda)y_2 \in Y$ and since $M$ is full row-rank, the pre-image of $\lambda y_1 + (1-\lambda)y_2$ must be unique. This concludes both directions of the proof.

\end{document}